\documentclass[12pt]{article}

\usepackage{amssymb, amsmath, amsthm}
\usepackage{graphicx,subfig,tikz}
\usepackage{cite}
\usepackage[left=2cm,right=2cm,top=2cm,bottom=2cm]{geometry}

\newtheorem{theorem}{Theorem}
\newtheorem{lemma}{Lemma}

\newtheorem{cor}{Corollary}
\newtheorem{prop}{Proposition}

\usepackage[colorlinks=true,linktocpage=true,allcolors=blue]{hyperref}

\def\be{\begin{equation}}
\def\ee{\end{equation}}
\def\bea{\begin{eqnarray}}
\def\eea{\end{eqnarray}}

\newcommand{\td}{\text{d}}

\title{ \bf{Uniqueness of supersymmetric AdS$_5$ black holes with $SU(2)$ symmetry}}

\author{James Lucietti\footnote{j.lucietti@ed.ac.uk}\; and Sergei G. Ovchinnikov\footnote{s.g.ovchinnikov@sms.ed.ac.uk} 
\\ \\ \small \sl School of Mathematics and Maxwell Institute for Mathematical Sciences, \\ \small \sl    University of Edinburgh, King's Buildings, Edinburgh, EH9 3JZ, UK }

\date{}

\begin{document}

\maketitle

\begin{picture}(0,0)(0,0)
\put(350, 240){}
\put(350, 225){}
\end{picture}

\begin{abstract} 
We prove that any  supersymmetric solution   to five-dimensional minimal gauged supergravity with $SU(2)$ symmetry, that  is timelike outside an analytic horizon, is a Gutowski-Reall black hole or its near-horizon geometry.  The proof combines a delicate near-horizon analysis with the general form for a  K\"ahler metric with cohomogeneity-1 $SU(2)$ symmetry.
We also prove that any timelike supersymmetric soliton solution to this theory, with $SU(2)$ symmetry and a nut or a complex bolt,  has a K\"ahler base with enhanced $U(1)\times SU(2)$ symmetry, and we exhibit a family of asymptotically AdS$_5/\mathbb{Z}_p$ solitons for $p \geq 3$ with a bolt in this  class. 
\end{abstract}

\newpage

\tableofcontents
\section{Introduction}

Black holes in Anti de Sitter (AdS) spacetime are of central importance in the AdS/CFT correspondence~\cite{Witten:1998zw}.   A fundamental problem in this context is to provide a microscopic derivation of the Bekenstein-Hawking entropy from the dual CFT.   It is expected this problem is more technically tractable for supersymmetric black holes, as is the case for their flat space counterparts~\cite{Strominger:1996sh, Breckenridge:1996is}. In recent years, there has been striking progress on the CFT side of this problem for AdS$_4$ black holes and most recently for AdS$_5$ black holes~\cite{Benini:2015eyy, Cabo-Bizet:2018ehj, Choi:2018hmj, Benini:2018ywd} (see also the review~\cite{Zaffaroni:2019dhb}). Naturally, a full resolution to this problem ultimately also requires a complete classification of black holes in AdS.

The classification of equilibrium black holes in AdS spacetime turns out to be even more nontrivial than in flat space. Indeed, very few general results in this direction are even known. One obvious source of complication arises because the conformal boundary at null infinity can be taken to be any Lorentzian metric giving rise to a large class of so-called asymptotically locally AdS spacetimes. However, even restricting to asymptotically globally AdS spacetimes, so the conformal boundary is the Einstein static universe, the classification problem remains poorly understood. Notably, even in  four-dimensions, analogues of the black hole uniqueness theorem for the Kerr solution are not known for AdS spacetimes. 

The Kerr-AdS solution and its higher-dimensional counterpart~\cite{Hawking:1998kw, Gibbons:2004uw} are the only known family of vacuum black holes in global AdS. In  $D=4,5$ gauged supergravity,  charged versions of the Kerr-AdS family are known~\cite{Kostelecky:1995ei, Cvetic:2005zi, Chong:2005hr, Wu:2011gq}, including supersymmetric black holes~\cite{Gutowski:2004ez, Gutowski:2004yv, Kunduri:2006ek}. On the mathematical side,  the major obstacle to constructing and classifying solutions is that the presence of a cosmological constant spoils integrability of the Einstein equations for stationary and axisymmetric solutions, which no longer reduce to a harmonic map.  Furthermore, on the physical side, it is even expected that the no-hair theorem is violated by the existence of rotating black holes with a single (corotating) Killing field, which may arise as the endpoint of a superradiant instability\cite{Kunduri:2006qa, Cardoso:2006wa, Dias:2015rxy}.

In this paper we will consider the classification of black holes in global AdS$_5$, under the additional assumption of supersymmetry. In particular, we will consider supersymmetric backgrounds of $D=5$ minimal gauged supergravity, which is the simplest theory that admits such solutions. Supersymmetry typically heavily constrains the possible backgrounds and therefore offers a setting where the classification of black holes may be tractable. 
The most general known asymptotically global AdS$_5$ black hole solution to $D=5$ minimal gauged supergravity was found by Chong, Cvetic, Lu, Pope (CCLP), and is a four parameter family of topologically $S^3$ black holes~\cite{Chong:2005hr}. It admits a two-parameter supersymmetric limit which carries two independent angular momenta $J_1, J_2$, electric charge $Q$ and a nonlinear constraint relating these (the mass is fixed by the BPS relation). This raises an obvious question:  are there other asymptotically global AdS$_5$ supersymmetric black holes in this theory?\footnote{The possibility of hairy supersymmetric black holes in other truncations of supergravity has recently been investigated~\cite{Markeviciute:2018yal, Markeviciute:2018cqs}.}

 It is instructive to compare to the analogous classification problem for asymptotically flat supersymmetric black holes in $D=5$ minimal (ungauged) supergravity.  For some time the only known solutions were the BMPV black hole~\cite{Breckenridge:1996is} and a black ring~\cite{Elvang:2004rt}.  Recently, new examples of supersymmetric holes in this context have been discovered: black holes with lens space topology (black lens)  and black holes with topological bubbles (2-cycles) in the domain of outer communication~\cite{Kunduri:2014iga, Kunduri:2014kja, Tomizawa:2016kjh, Horowitz:2017fyg, Breunholder:2018roc}.  Furthermore, a classification theorem has been established for such solutions under the additional assumption of biaxial symmetry~\cite{Breunholder:2017ubu}. 
 This reveals a rich moduli space of black holes with $S^3$, $S^2\times S^1$ and lens space $L(p,1)$ horizon topology that are roughly enumerated by the number of black holes and number of 2-cycles  in the exterior region.  It is natural to wonder if there are asymptotically AdS$_5$ solutions of this kind. 
 
 Supersymmetric black holes are extremal and therefore admit a well-defined near-horizon geometry that itself is a solution.  The classification of supersymmetric near-horizon geometries with compact horizons has been achieved in this theory~\cite{Kunduri:2006uh,Grover:2013hja}. It shows that the most general  regular solution is locally isometric to the near-horizon geometry of the CCLP black hole. In particular, this proves the horizon must have $S^3$ or lens space topology.  Thus there are no supersymmetric black rings in AdS$_5$.\footnote{It is worth noting that in the more general $U(1)^3$ gauged supergravity one cannot rule out supersymmetric black rings as there do exist supersymmetric near-horizon geometries with ring topology~\cite{Kunduri:2007qy}. } This is striking as supersymmetric black rings in flat space do exist.  
Nevertheless, comparison to the classification in flat space still leaves the following questions: Are there black holes with lens space topology in AdS$_5$? Are there black holes with 2-cycles in the exterior  in AdS$_5$?  Are there multi-black holes in AdS$_5$? 

We emphasise that these questions are not in contradiction with the near-horizon uniqueness theorem, since that does allow for lens space topology horizons and does not say anything about the topology of the exterior region. Indeed,  as the classification of supersymmetric black holes in flat space shows, there can be many solutions with (locally) isometric near-horizon geometries. Furthermore, it is worth noting that the supersymmetric CCLP solution reduces to the BMPV solution in the limit of vanishing cosmological constant. For these reasons, it is therefore far from clear that CCLP is the only AdS$_5$  black hole solution of $D=5$ minimal gauged supergravity. 

Unfortunately, even the classification of supersymmetric solutions in AdS is difficult. For $D=5$ minimal gauged supergravity this reduces to a finding a 4d K\"ahler geometry that solves a complicated 4th order nonlinear PDE for its curvature~\cite{Gauntlett:2003fk,Cassani:2015upa}. Therefore a general local classification of solutions is not currently available.  Hence it is natural to seek further symmetry assumptions that are compatible with spacetimes in this class. Asymptotically globally AdS$_5$ spacetimes  have an $SO(4)$ rotation group acting on the spatial $S^3$ at infinity. If one assumes spherical $SO(4)$ symmetry the solution must be static and hence is a supersymmetric limit of the Reissner-Nordstr\"om-AdS$_5$ solution which is nakely singular. A natural class to investigate is solutions that are invariant under an abelian $U(1)^2\subset SO(4)$ rotational symmetry, which in particular contains the CCLP solution. This would also be the AdS analogue of the aforementioned classification recently obtained for asymptotically flat supersymmetric black hole solutions~\cite{Breunholder:2017ubu}. Unfortunately, this appears to be a difficult problem in toric K\"ahler geometry.

In five-dimensions there is another notion of axisymmetry.  Namely, solutions that preserve an $SU(2)$ subgroup of $SO(4) \cong (SU(2)\times SU(2)) /\mathbb{Z}_2$. In fact, the first example of a supersymmetric black hole in AdS$_5$ was found by Gutowski and Reall~\cite{Gutowski:2004ez}, which corresponds to a one-parameter subfamily of CCLP characterised by having equal angular momenta $J_1=J_2$ and an enhanced rotational symmetry $U(1)\times SU(2)$, where $SU(2)$ acts with 3d orbits. 
A simple question therefore presents itself: can we classify all supersymmetric solutions that also admit such an $SU(2)$-symmetry? This reduces to classifying K\"ahler metrics with a cohomogeneity-1 $SU(2)$ symmetry, which is an ODE problem and therefore much more tractable.
Our main result is given by the following  uniqueness theorem.

\begin{theorem} \label{theorem}
\label{thm}Any supersymmetric and $SU(2)$-symmetric solution to five-dimensional minimal gauged supergravity, that is timelike  outside an analytic horizon with compact cross-sections, must be a Gutowski-Reall black hole or its near-horizon geometry. 
\end{theorem}

 We believe this is the first uniqueness theorem for  black holes in  AdS in  dimension $D\geq 4$, supersymmetric or otherwise.\footnote{Naturally, in three-dimensions black hole uniqueness results are known~\cite{Li:2013pra}.}  We emphasise that no global assumptions on the spacetime are made other than the supersymmetric Killing field is timelike outside the horizon. In particular, our proof does not make any assumptions about the asymptotics and therefore also rules out the existence of asymptotically locally AdS$_5$ supersymmetric black holes in this class (other than quotients of the Gutowski-Reall solution). In fact, a class of supersymmetric black holes with $SU(2)$ symmetry that are asymptotically locally AdS$_5$ with squashed $S^3$ spatial boundary metrics have been constructed numerically~\cite{Blazquez-Salcedo:2017kig, Blazquez-Salcedo:2017ghg, Cassani:2018mlh, Bombini:2019jhp}. Our results show that these solutions do not possess smooth horizons.  In fact, we will show  that the horizons of such solutions are $C^1$ but not $C^2$.

The structure of the proof is as follows. First we show that any K\"ahler metric with a cohomogeneity-1 $SU(2)$ symmetry is described by a simple system of ODEs, which generalises that of Dancer and Strachan for K\"ahler-Einstein metrics with $SU(2)$ symmetry~\cite{Dancer:1993aw}. 
On the other hand, by a delicate near-horizon analysis we show that a horizon corresponds to a conical singularity of the K\"ahler base. This provides boundary conditions for the ODE system that governs the K\"ahler base, which imply that it must have an enhanced  $U(1)\times SU(2)$ symmetry everywhere outside the horizon. Finally, we show that the classification of $U(1)\times SU(2)$-invariant supersymmetric solutions reduces to a single 5th order nonlinear ODE and we are able to find all solutions to this that correspond to an analytic horizon. 

As part of our analysis, we also determine the general local form for a timelike supersymmetric solution to $D=5$ minimal gauged supergravity with a K\"ahler base with a  cohomogeneity-1 $SU(2)$ symmetry.  Furthermore, we derive the boundary conditions required for a K\"ahler metric in this class to contain a nut or bolt and use this to prove that any strictly timelike supersymmetric soliton spacetime with a nut or a complex bolt must have  enhanced $U(1)\times SU(2)$ symmetry. As an example,  we construct an explicit class of solitons with a bolt that are asymptotically locally AdS$_5$ with spatial $S^3/\mathbb{Z}_p$ boundary for $p\geq 3$.

This paper is organised as follows.  In Section \ref{sec:gen} we review the local form of supersymmetric backgrounds of $D=5$ minimal gauged supergravity and then examine the consequences of $SU(2)$ symmetry. In Section \ref{sec:su2kahler} we derive the general form for a K\"ahler metric with cohomogeneity-1 $SU(2)$ symmetry and establish a symmetry enhancement result for such metrics; this section is written in a self-contained way as it may be of independent interest. In Section \ref{sec:susysu2solns} we determine the general local form for a supersymmetric solution with $SU(2)$-symmetry and discuss the known examples of black hole and soliton solutions. In Section \ref{sec:BH} we analyse supersymmetric black holes with $SU(2)$ symmetry and prove the black hole uniqueness theorem. We close with a Discussion of our results and provide an Appendix.

\section{Supersymmetric solutions of gauged supergravity}
\label{sec:gen}

\subsection{General local classification}
The general local form for supersymmetric solutions $(M, g, F)$ to five-dimensional minimal gauged supergravity was determined in~\cite{Gauntlett:2003fk} (we work in the conventions of~\cite{Gutowski:2004yv}). We will briefly recall it here. For timelike solutions, which are defined by the supersymmetric Killing field $V$ being timelike, the metric can be written as
\be\label{metricform}
g = -f^2( \td t+\omega)^2+ f^{-1} h\; ,
\ee
where $V=\partial_t$, the metric $h$ on the orthogonal base $B$ is K\"ahler, and $f$ and $\omega$ are a function and 1-form on $B$. The Maxwell field takes the form
\be
F = \frac{\sqrt{3}}{2} \td (f( \td  t+\omega)) - \frac{1}{\sqrt{3}} G^+ - \frac{\sqrt{3}}{\ell f} J \; ,  \label{maxwell}
\ee
where $G^\pm = \tfrac{1}{2}f ( \td\omega \pm \star_B \td\omega )$, $J$ is the K\"ahler form of $(B, h)$ and the orientation on $B$ is such that $J$ is anti-self dual (ASD), i.e. $\text{vol}_B= - \tfrac{1}{2} J \wedge J$.

Given a K\"ahler base, the following are completely fixed in terms of its curvature
\be
f^{-1} = - \frac{\ell^2}{24} R, \qquad G^+ = -\frac{\ell}{2} (\mathcal{R}- \tfrac{1}{4} J R) \; ,  \label{fGp}
\ee
where $\mathcal{R}_{ab}= \tfrac{1}{2} R_{ab}^{~~cd} J_{cd}$ is the Ricci form and $R_{abcd}$ is the Riemann tensor of $B$. However, as clarified in~\cite{Figueras:2006xx, Cassani:2015upa}, the K\"ahler base is not free to be chosen.

First, we recall that any K\"ahler surface, with ASD K\"ahler form $\Omega^3:= J$, admits a complex $(2,0)$ form $\Omega^1+ i \Omega^2$, with $\Omega^1, \Omega^2$ real ASD 2-forms,  which satisfy the quaternion algebra
\be
(\Omega^i)^a_{~c} (\Omega^j)^c_{~b} = -\delta^a_{~b}\delta_{ij} +\epsilon_{ijk}( \Omega^k)^a_{~b}  \;  , \label{quaternion}
\ee
for $i=1, 2,3$, and the differential equations
\be
\nabla_a \Omega^1_{ bc} = P_a \Omega^2_{ bc}, \qquad \nabla_a \Omega^2_{ bc}=- P_a \Omega^1_{ bc} \; ,  \label{X2X3P}
\ee
where $\nabla$ is the metric connection of $h$ and $P$ is the potential for the Ricci form $\mathcal{R}=\td P$. Observe that there is a local $O(2)$ freedom $\Omega^i \to O^i_{~j} \Omega^j$ which preserves the K\"ahler form $\Omega^3$ and generates gauge transformations of the Ricci form potential $P$.

In particular, since the $\Omega^i$ give a basis for ASD 2-forms we can expand
\be
G^-= \frac{\ell}{2 R}  \lambda_i \Omega^i  \; ,\label{Gm}
\ee
where 
\be
\lambda_3= \tfrac{1}{2} \nabla^2 R+ \tfrac{2}{3} R_{ab}R^{ab}- \tfrac{1}{3} R^2  \; .
\ee
Then, as shown in~\cite{Cassani:2015upa}, the integrability condition for
\be 
\td \omega= f^{-1}( G^++G^-) \; , \label{dom}
\ee
fixes $\lambda_1+i \lambda_2$ up to an antiholomorphic function on $B$ and
requires that the K\"ahler metric satisfies the following complicated 4th order PDE for its curvature,
\be
\nabla^c \Xi_c=0  \; ,\label{PDE}
\ee
where 
\be
\Xi_c:= \nabla_c \left(\frac{1}{2} \nabla^2 R+\frac{2}{3} R_{ab}R^{ab}- \frac{1}{3} R^2 \right)+ R_{cb}\nabla^b R \; .
\ee
Conversely, given a K\"ahler metric $h$ which satisfies  (\ref{PDE}) a supersymmetric solution can be reconstructed by solving (\ref{dom}) for $\omega$ and $f$ is determined by (\ref{fGp}).

To summarise, the classification of  timelike supersymmetric  solutions in this theory reduces to the classification of K\"ahler metrics that satisfy (\ref{PDE}).  It is worth noting that any K\"ahler-Einstein metric  satisfies (\ref{PDE}) and therefore gives a supersymmetric solution  with $f=1$ and $\td \omega$ is an ASD 2-form (normalised so $R_{ab}= -6 \ell^{-2} h_{ab}$).

\subsection{Solutions with $SU(2)$ symmetry}

We will now assume that $(M,g,F)$ also admits a $G$-symmetry in the following sense: (i) there is an  isometry group $G$ with spacelike orbits; (ii) the supersymmetric Killing field $V$ is complete and invariant under $G$ so that the spacetime isometry group is $\mathbb{R}_t\times G$ where $\mathbb{R}_t$ is generated by $V$; (iii) the Maxwell field is invariant under  $G$. 

We start by deducing the constraints on the data for a timelike solution $(f, \omega, h)$ imposed by such a $G$-symmetry. These are summarised by the the following result.
\begin{lemma}
\label{lem:Gsym}
A timelike supersymmetric  solution $(f, \omega, h)$ that admits a $G$-symmetry (as above), has a K\"ahler base metric $h$ with a holomorphic $G$-symmetry and $f, \omega$ are $G$-invariant.
\end{lemma} 
\begin{proof}
Under these assumptions, it follows that $f^2=-g_{\mu\nu}V^\mu V^\nu$ is $G$-invariant. Therefore,  the K\"ahler metric on the orthogonal base,
\be
h_{\mu\nu}= f \left( g_{\mu\nu}-  \frac{V_\mu V_\nu}{|V|^2} \right) \; ,
\ee
is also $G$-invariant. Furthermore, since $V_\mu \td x^\mu = -f^2(\td t+\omega)$ is $G$-invariant, the gauge freedom $t\to t+\lambda, \omega\to \omega- \td \lambda$, where $\lambda$ is a function on the base, can be used to ensure $\omega$ is a $G$-invariant 1-form on the base space. Finally, notice that invariance of the Maxwell field (\ref{maxwell}) also implies that the K\"ahler form $J$ is $G$-invariant so the $G$-action is holomorphic.
\end{proof}

Now we will further restrict to $G$ with generic 3d orbits, where $G$  is $SU(2)$ or $U(1)\times SU(2)$.  Let $L_i$ and $R_i$,  $i=1,2,3$, denote the generators of the left and right action on $SU(2)$ respectively. In  Appendix \ref{app:su2} we recall the standard formulas of $SU(2)$ calculus and define our conventions. Without loss of generality, we assume the isometry group  $G=SU(2)$  is generated by the right-action vectors $R_i$ (i.e. the left-invariant vector fields). The right-invariant 1-forms $\sigma_i$, which are dual to $L_i$, obey
\be
\td \sigma_i = -\tfrac{1}{2} \epsilon_{ijk} \sigma_ j\wedge \sigma_k  \; .  \label{MCsu2}
\ee
Then, any $SU(2)$-invariant metric, in the sense $\mathcal{L}_{R_i} h=0$,  locally can  be written in the form
\be
h= \td \rho^2+ h_{ij}(\rho)\sigma_i \sigma_j \; ,   \label{su2metric}
\ee
where $h_{ij}:= h(L_i, L_j)$ is a positive-definite Gram-matrix and  $\rho$ is a local coordinate orthogonal to the $SU(2)$-orbits. It is worth noting that there is a global $SO(3)$ freedom acting on the right-invariant forms $\sigma_i \to R_{ij} \sigma_j$ where $R\in SO(3)$ is a constant rotation. In general, this can be used to diagonalise $h_{ij}(\rho)$ only at a point. We will refer to $h_{ij}(\rho)$ as diagonalisable if it can be diagonalised for all $\rho$ (in an appropriate domain).

The most general $SU(2)$-invariant 1-form can be written as
\be
\omega= \omega_i(\rho) \sigma_i  \; ,  \label{su2omega}
\ee
where we have exploited the gauge freedom in the definition of $\omega$ mentioned above to fix $\omega_\rho=0$. Clearly, invariance of the function $f$ implies that it can only depend on $\rho$.

Lemma \ref{lem:Gsym} also shows that that the K\"ahler form is $SU(2)$-invariant.  It is easy to show that the most general $SU(2)$-invariant closed 1-form must take the form
\be
J= \td ( g_i(\rho)\sigma_i)  \; .   \label{su2J}
\ee
Furthermore, the condition that $J$ be ASD is equivalent to
\be
g_i'= \frac{h_{ij} g_j}{\sqrt{\det h}}  \; , \label{JASD}
\ee
where we choose the orientation of the base to be $\td \rho \wedge \sigma_1\wedge \sigma_2 \wedge \sigma_3$.
We also require that $J$ defines an almost complex structure, i.e. $J^a_{~c} J^c_{~b} =-\delta^a_{~b}$.  It is straightforward to show that this condition reduces to
\be
h_{ij} g_i g_j= \det h \;,\label{JACS}
\ee
upon use of the ASD condition.

We require the necessary and sufficient conditions for (\ref{su2metric}) to be a K\"ahler metric with  K\"ahler form (\ref{su2J}). Imposing that $J$ is a closed ASD 2-form that defines an almost complex structure as above is not sufficient in general. One must also require that $J$ is an integrable complex structure. For a diagonal $h_{ij}(\rho)$ these conditions were obtained by Dancer and Strachan in a study of K\"ahler-Einstein metrics with $SU(2)$-symmetry~\cite{Dancer:1993aw}.  For Einstein metrics it turns out one can show that $h_{ij}(\rho)$ can always be diagonalised. However, in general this need not be the case, therefore we must consider the most general nondiagonal $h_{ij}(\rho)$.  It is convenient to perform this calculation in an orthonormal frame. We will present this calculation in Section \ref{sec:su2kahler}.

Now consider $G=U(1) \times SU(2)$ where $U(1)$ is generated by a subgroup of the left-action and $SU(2)$ by the right-action as above. This may be viewed as a special case of the $G=SU(2)$ case above that is also invariant under a $U(1)$ of the left-action, which without loss of generality we take to be generated by $L_3$. Then the general form for an invariant metric takes the form (\ref{su2metric}) with the further restriction $\mathcal{L}_{{L}_3}h=0$, which implies
\be
h= \td \rho^2+ a(\rho)^2( \sigma_1^2+ \sigma_2^2)+ c(\rho)^2 \sigma_3^2  \label{hsu2u1}  \; .
\ee
This is a special case of the diagonal metric.  Invariance of the K\"ahler form $\mathcal{L}_{L_3}J=0$ implies $g_1=g_2=0$ and therefore (\ref{JASD}), (\ref{JACS}) reduce to
\be
c= (a^2)', \qquad J = \td (a(\rho)^2 \sigma_3) \; ,   \label{Ju1su2}
\ee
where we have fixed an overall sign in $J$.
In this case, it turns out these conditions are sufficient to ensure $J$ is an integrable complex structure and hence $(h,J)$ is a K\"ahler structure.
Indeed, this is precisely the class of K\"ahler bases considered in~\cite{Gutowski:2004yv}.
The moment map defined by the $U(1)$-symmetry generated by ${L}_3$ is
\be
\td \mu= - \iota_{{L}_3}J= a^2  \; .
\ee
It is natural to define a new coordinate from the moment map by  $r:=2a(\rho)$ in terms of which 
\bea
&&h= \frac{\td r^2}{V(r)} + \frac{r^2}{4} (\sigma_1^2+\sigma_2^2) + \frac{r^2 V(r)}{4} \sigma_3^2 \; , \label{EHgauge} \\
&&J= \td \left(\tfrac{1}{4}  r^2 \sigma_3 \right) \; ,
\eea
where $V(r):= 4 a'(\rho)^2$. We remark that $a(\rho)$ must be a nonconstant function, otherwise $c(\rho)=0$ and the metric is degenerate, so one can always introduce the local coordinate $r$. Therefore this represents the most general $U(1)\times SU(2)$ invariant K\"ahler structure.  Finally, note that the most general $U(1)\times SU(2)$ invariant 1-form can be written as
\be
\omega= \omega_3(r) \sigma_3  \; , \label{omsu2u1}
\ee
where as above we have used the gauge freedom in its definition to set $\omega_r=0$.

For orientation, it is worth noting how global AdS$_5$ is described. The K\"ahler base for this is the Bergmann metric, which is an Einstein metric with $U(1)\times SU(2)$ symmetry  given by $a= \frac{\ell}{2} \sinh (\rho /\ell)$~\cite{Gauntlett:2003fk}. In terms of the $r$-coordinate this corresponds to
\be
V= 1+ \frac{r^2}{\ell^2}  \label{VAdS}
\ee
and the rest of the data for AdS$_5$ is simply
\be
f=1, \qquad \omega = \frac{r^2}{2\ell}\sigma_3  \; . \label{fomAdS}
\ee

\section{Classification of K\"ahler metrics with $SU(2)$ symmetry}
\label{sec:su2kahler}

In this section we will derive the general form of a cohomogeneity-1 K\"ahler metric with $SU(2)$ symmetry.  We  emphasise that we do not assume the metric on the surfaces of transitivity is diagonalisable and therefore obtain a generalisation of the conditions derived by Dancer and Strachan for diagonal metrics~\cite{Dancer:1993aw}.  We will also obtain some results on the global analysis of such geometries. This section may be of independent interest, so we give a self-contained presentation. We give our $SU(2)$ conventions in Appendix \ref{app:su2}.

For clarity, in this section we denote the K\"ahler 2-form by $\Omega$ and the complex structure by $J$ (in the rest of the paper we denote them both by $J$ since we work in conventions  compatible with raising and lowering indices with the K\"ahler metric, i.e. $\Omega_{ab}= h_{ab} J^b_{~c}$ etc).

\subsection{Local geometry}

It is convenient to introduce an $SU(2)$-invariant orthonormal frame $e^0= \td \rho$,  $e^i= E^i_{~j}(\rho) \sigma_j$, where $E^i_{~j}$ is an invertible matrix, so that a general $SU(2)$-invariant metric (\ref{su2metric}) is 
\be
h= (e^0)^2 +e^i e^i,
\ee
oriented so $e^0\wedge e^1\wedge e^2\wedge e^3$ is positive.  Then a basis of ASD 2-forms $\Omega^i$  is given by
\be
\Omega^i = e^0\wedge e^i -\frac{1}{2} \epsilon_{ijk} e^j \wedge e^k  \; . \label{OmASD}
\ee
Note that these satisfy the quaternion algebra (\ref{quaternion}) and are manifestly $SU(2)$-invariant.  
The frame $e^i$ is defined only up to a local $O(3)$ transformation $e^i \to O^i_{~j}(\rho) e^j$ where $O\in O(3)$.  Under this the basis of ASD 2-forms transforms as $\Omega^i\to O^i_{~j} \Omega^j$ provided $\det O=1$. Therefore, without loss of generality we may use this freedom to fix the K\"ahler form to be\footnote{Dancer and Strachan take $\Omega= A_i(\rho) \Omega_i$ since they work in a frame in which the metric is diagonal and therefore can't use the $SO(3)$ symmetry to rotate the K\"ahler form to (\ref{kahler}).}
\be
\Omega= \Omega^3  \; .  \label{kahler}
\ee
Furthermore, $\Omega^1, \Omega^2$ must then satisfy (\ref{X2X3P}).

Next, it is convenient to introduce the following frame
\be
e^0= \td \rho, \qquad e^1=a \sigma_1+ b_1 \sigma_2+ c_1 \sigma_3, \qquad e^2 = b \sigma_2+ c_2 \sigma_3, \qquad e^3= c \sigma_3  \; ,  \label{nonDframe}
\ee
where $a, b, c, b_1, c_1, c_2$ are functions of $\rho$, which parameterises the most general $SU(2)$-invariant metric provided $abc \neq 0$ (to see this note that any positive-definite symmetric matrix, such as the metric $h_{ij}$, can be factored as $h_{ij}=(E^T E)_{ij}$ where $E$ is upper triangular).   We are now ready to state the main result of this section.

\begin{theorem}\label{th:kahler}
 The most general K\"ahler metric with a cohomogeneity-1 $SU(2)$ symmetry can be written in the frame (\ref{nonDframe}) where $b_1=c_1=c_2=0$,  so in particular is diagonal, and
 \bea
 &&2 bc a'=  b^2-a^2+c^2 \; , \label{ap}\\
  &&2ac b'= a^2-b^2+c^2 \label{bp}  \; .
  \eea
 The K\"ahler form  is simply
\be
\Omega=\td (a b \sigma_3).  \label{su2Om}
\ee
\end{theorem}

\begin{proof}
First we will work in a general orthonormal frame $e^0= \td \rho, e^i = E^i_{~j}(\rho) \sigma_j$ and denote the dual basis by $X_0=\partial_\rho, X_i= (E^{-1})^{j}_{~i} L_j$, where $L_i$ are the right-invariant vector fields dual  to $\sigma_i$.  As argued above, we may assume the K\"ahler form $\Omega$ is given by (\ref{kahler}), that is,
\be
\Omega= e^0\wedge e^3- e^1\wedge e^2  \; .
\ee
Then the action of the almost complex structure $J$ on any vector $X$ can be found using $J X=- h^{-1}( \iota_X \Omega, \cdot)$, which gives
\be
J X_0 = -X_3 , \qquad J X_1 =X_2 , \qquad JX_2 =-X_1  \qquad JX_3 = X_0 \; .
\ee
We wish to impose that $J$ is a complex structure, i.e. that it is integrable.  A convenient method to do this is as follows~~\cite{Dancer:1993aw}.  Let $\chi_0, \chi_1$ be a basis of $(1,0)$-forms, i.e. $J \chi_{0,1} = i \chi_{0,1}$.  Then require that $[\chi_0, \chi_1]$ is also a $(1,0)$-form.  We choose
\be
\chi_0 = X_0 - i J X_0, \qquad \chi_1= X_1- i JX_1 \; ,
\ee
which are indeed always linearly independent.

Now, without loss of generality we parameterise the frame by (\ref{nonDframe}). First note that the requirement that the K\"ahler form $\Omega$ is closed is equivalent to
\be
c= (ab)' \; ,  \label{closedOm}
\ee
as well as certain first order ODEs for $c_1, c_2$ which we will not need.
The dual vectors for our orthonormal frame are
\bea
&&X_0= \partial_\rho, \qquad X_1= \frac{1}{a} L_1, \qquad X_2= \frac{1}{b} \left( L_2- \frac{b_1}{a} L_1 \right), \nonumber \\
&& X_3=\frac{1}{c} \left( L_3- \frac{c_2}{b} L_2 - \frac{1}{a} \left(c_1-\frac{b_1 c_2}{b} \right)L_1\right) \; .
\eea
 A computation gives
\be
[\chi_0, \chi_1]= k_i L_i
\ee
where
\be
k_1= -\frac{a'}{a^2}-\frac{1}{bc}+ i \left(\frac{b_1}{ab}\right)', \qquad k_2= \frac{i}{ac}+\frac{i b'}{b^2}-\frac{b_1}{abc}, \qquad k_3= \frac{ic_2-c_1}{abc}\; .
\ee
Since the $\rho$ component of $[\chi_0, \chi_1]$ vanishes, the integrability condition  $J[\chi_0, \chi_1]=i [\chi_0, \chi_1]$ is equivalent to
\be
a k_1+ b_1 k_2=i b k_2, \qquad k_3=0  \;.
\ee
Using the explicit values for $k_i$, one finds the integrability of $J$  reduces to 
\be
c( b a'-ab')+ a^2-b^2+b_1^2=0, \qquad  ac b_1'= b_1 ( c a'- 2 b ), \qquad c_1=c_2=0  \; .  \label{intconds}
\ee
Finally, note that $c_1=c_2=0$ satisfy the ODEs mentioned above that arise from the closure of $\Omega$, so this condition is now  equivalent to (\ref{closedOm}).  

We can express these conditions is a slightly more convenient form. Solving (\ref{closedOm}) and the first equation in (\ref{intconds}) for $a',b'$ gives
 \bea
 &&2 bc a'=  b^2-a^2+c^2- b_1^2,\\
 && 2ac b'= a^2-b^2+c^2 + b_1^2,
\eea
and substituting into the second equation in (\ref{intconds}) gives 
\be
2 abc b_1'= - b_1( a^2+3 b^2-c^2 +b_1^2)  \; . \label{b1p}
\ee
Now, recall that any $SU(2)$-invariant metric (\ref{su2metric}) can be diagonalised at a point by a  constant rotation $\sigma_i \to R_{ij} \sigma_j$ where $R\in SO(3)$. In terms of the frame (\ref{nonDframe}) this means we may always arrange $b_1=c_1=c_2=0$ for some value of $\rho$ in our domain. On the other hand (\ref{b1p}) immediately shows that if $b_1=0$ for some $\rho$ then $b_1=0$ for all $\rho$ in the domain where $abc \neq 0$.  We deduce that without loss of generality we may set $b_1=0$.

Therefore, we have shown that any K\"ahler metric defined by (\ref{nonDframe}) with K\"ahler form (\ref{kahler}), can be arranged to be diagonal, i.e. $b_1=c_1=c_2=0$, where  (\ref{ap}), (\ref{bp}) satisfied. It is easy to show that the K\"ahler form can be written as (\ref{su2Om}).
\end{proof}

It is worth  remarking that our theorem  reduces to that of Dancer and Strachan for Einstein metrics with a cohomogeneity-1 $SU(2$) symmetry~\cite{Dancer:1993aw}.  Interestingly, we have shown that even without the Einstein condition, one can always choose a diagonal metric. Furthermore, setting $a^2=b^2$ it is easy to see it reduces to the $U(1)\times SU(2)$ invariant case described by (\ref{hsu2u1}) and (\ref{Ju1su2}). On the other hand, setting $2ab c'= a^2+b^2-c^2$ gives a hyper-K\"ahler metric;  indeed, this is equivalent to $a=(bc)'$,  $b=(ac)'$ and (\ref{closedOm}), where $\Omega^1=\td (bc \sigma_1)$, $\Omega^2=\td ( ac \sigma_2)$ also define integrable complex structures.

It is also worth noting that the ODE system in Theorem \ref{th:kahler} is invariant under interchange of $a$ and $b$. This corresponds to the (orientation-preserving) discrete symmetry  $\sigma_1 \to \sigma_2$, $\sigma_2\to -\sigma_1$ of diagonal metrics with K\"ahler form (\ref{su2Om}).

For convenience we record  certain curvature formulas for this class of K\"ahler metrics. The Ricci scalar is 
\be
R= -\frac{1}{ a^2 b^2 c^2} \left( a^4+b^4- (a^2+b^2) c^2 +4a b c^2 c' +2 a^2 b^2(-1+c c'')  \right)  \; , \label{Rscalar}
\ee
where we have eliminated first derivatives of $a,b$ using (\ref{ap}), (\ref{bp}). The Ricci form is given by the potential
\be
P= \left(\frac{a^2+b^2-c^2-2 ab c'}{2ab} \right)\sigma_3  \; . \label{Rform}
\ee
It is worth noting that a convenient trick for computing the Ricci form is to invert (\ref{X2X3P}) for the potential $P$.  The Einstein condition $R_{ab}= \Lambda h_{ab}$ is equivalent to the Ricci form satisfying $\mathcal{R}_{ab}= \Lambda \Omega_{ab}$. Comparing (\ref{su2Om}) and (\ref{Rform}) immediately implies that the Einstein condition is equivalent to $2 ab c' = a^2+b^2-c^2-2 \Lambda a^2 b^2$ in agreement with~\cite{Dancer:1993aw}.

\subsection{Symmetry enhancement}

Here we prove a  general symmetry enhancement result for K\"ahler metrics with $SU(2)$ symmetry under appropriate boundary conditions. This will be useful for our later analysis.

To this end, it is useful to note that in terms of the function
\be
T := \frac{b}{a}  \; , \label{Tdef}
\ee
the ODE system in Theorem \ref{th:kahler} implies 
\bea
c T' &=& 1-T^2 \; . \label{Teq}
\eea
Observe that in terms of this the $U(1)\times SU(2)$-invariant  case $a^2=b^2$ is simply $T^2=1$.
This ODE  allows us to prove the following elementary result.

\begin{lemma} \label{lem:sym}
Consider a K\"ahler metric with $SU(2)$ symmetry as in Theorem \ref{th:kahler}. Suppose $a,b,c$ are all positive and $C^1$ for $\rho>\rho_0$, that $abc=0$ at $\rho=\rho_0$, and that $\lim_{\rho \to \rho_0^+} T $ exists.  
If $T=1$  at $\rho=\rho_1>\rho_0$ then $T=1$ for all $\rho>\rho_1$.  In particular, if $T=1$ at $\rho=\rho_0$ then $T=1$ for $\rho>\rho_0$. 
\end{lemma}

\begin{proof} Under the stated assumptions $T$ is positive and $C^1$ for $\rho>\rho_0$, whereas $T$ is $C^0$ at $\rho=\rho_0$. 

For the first part, we are given that $T=1$ at $\rho=\rho_1$. 
 If $T>1$ for some $\rho=\rho_2>\rho_1$ then (\ref{Teq}) implies $T$ is monotonically decreasing so that $T(\rho_1)>T(\rho_2)>1$ which is a contradiction. On the other hand, if $T<1$ at $\rho=\rho_2>\rho_1$ then (\ref{Teq}) implies $T$ is monotonically increasing so that $T(\rho_1)< T(\rho_2)<1$ which is again a contradiction. Therefore $T=1$ for all $\rho>\rho_1$ as claimed.

For the second part, we are given $T=1$ at $\rho=\rho_0$. 
Now pick a point $\rho_*>\rho_0$ and suppose $T>1$ ($T<1$) at $\rho=\rho_*$; then (\ref{Teq}) implies $T$ is monotonically decreasing (increasing) so in particular $T(\rho_0+\epsilon)>T(\rho_*)>1$ ($T(\rho_0+\epsilon)<T(\rho_*)<1$) for small enough $\epsilon>0$, so taking the limit $\epsilon \to 0$ we deduce $T(\rho_0)>T(\rho_*)> 1$ ($T(\rho_0)< T(\rho_*)<1$) which is a contradiction. Therefore we must have $T=1$ at $\rho=\rho_*$ and since this was an arbitrary point we have shown $T=1$ for all $\rho>\rho_0$.
\end{proof}

The above result shows that if one has an enhanced $U(1)\times SU(2)$ symmetry at a point, then it has this enhanced symmetry at all points in our domain.

\subsection{Nuts and bolts}
\label{sec:NB}

In this section we will determine appropriate boundary conditions in order to obtain complete K\"ahler metrics. Relative to the frame (\ref{nonDframe}), 
\be
\det h= a^2b^2 c^2
\ee
and therefore $a^2b^2c^2>0$ ensures the $h$ is a smooth invertible metric. We  characterise the possible boundaries where $\det h=0$ in the following.

\begin{lemma} \label{lem:BC} Suppose $a,b,c$ are $C^1$ at $p$. Then $\det h=0$ at $p$ is equivalent to one of the following  conditions at $p$:
\bea
&&\text{nut}: \qquad a=b=c=0   \\
&&\text{bolt I} : \qquad    c=0,\quad  a = \pm b \neq 0     \\
&&\text{bolt II}: \qquad    a=0, \quad b=\pm c \neq 0 
\eea
The orbits of $SU(2)$ through such $p$ are either $0$-dimensional (nut) or $2$-dimensional (bolt).
\end{lemma}

\begin{proof} 
Clearly $\det h=0$ at $p$ is equivalent to $abc=0$ at $p$. Thus suppose $abc=0$ at $p$.
We need only rule out the possibility of precisely two of the three $a,b,c$ vanishing, since the nut and bolt cases above cover the remaining cases as we show below. This can happen in two inequivalent ways: either $a=b=0$,  $a=c=0$ (the case $b=c=0$ is equivalent to the latter under interchange of $a$ and $b$).  In the first case (\ref{ap}), (\ref{bp}), imply $c=0$ at $p$. In the second case (\ref{ap}), (\ref{bp})  implies $c=0$ at $p$. Thus in any case, we get a nut.  

Now consider the different types of bolt. If $c=0$ then (\ref{ap}), (\ref{bp})  implies $a^2=b^2$ or a nut; if $a=0$ then (\ref{ap}), (\ref{bp})  implies $b^2=c^2$ or a nut (the case $b=0$ is equivalent to the latter under interchange of $a$ and $b$).
\end{proof}

In order to obtain complete K\"ahler metrics we require that the nut and bolt conditions correspond to coordinate singularities and that the metric extends smoothly at these points.  First consider the nut case. Then smoothness requires at the very least that $a, b, c$ are all proportional to $\rho$. This means that the metric approaches a K\"ahler cone over an $SU(2)$-invariant space. The possible K\"ahler geometries near such a nut are given by simply Taylor expanding $a,b,c$ around $\rho=0$ and solving the ODE system (\ref{ap}), (\ref{bp}). This is summarised by the following.

\begin{lemma}
Any K\"ahler metric with cohomogeneity-1 $SU(2)$ symmetry with a nut at $\rho=0$ is given by
\be
a=\alpha \rho+ O(\rho^2), \qquad b= \alpha\rho+O(\rho^2), \qquad c= 2 \alpha^2\rho+ O(\rho^2), \label{nutexp}
\ee
where $\alpha>0$ is a constant. Furthermore, the leading term
\be
h= \td \rho^2 + \rho^2 \left(  \alpha^2 (\sigma_1^2+ \sigma_2^2) + 4 \alpha^4 \sigma_3^2 \right) \; ,
\ee
is an exact K\"ahler  metric and is the most general K\"ahler cone with $SU(2)$-symmetry.

The metric is smooth at the nut if and only if $\alpha=1/2$ and the higher order terms in $a^2, b^2, c^2$ are smooth functions of $\rho^2$, in which case the K\"ahler form extends smoothly at the nut. In this case the space around the nut is diffeomorphic to $\mathbb{R}^4$ (the cone is the flat metric on $\mathbb{R}^4$).
\end{lemma}

It is worth noting that the K\"ahler cone geometry in the above lemma is precisely the K\"ahler base of the near-horizon geometry of the Gutowski-Reall black hole (if $\alpha>1/2$) and therefore it satisfies (\ref{PDE})  (in fact $\Xi=0$).
 We may now combine this lemma with Lemma \ref{lem:sym} to deduce the following symmetry enhancement result.

\begin{prop}
\label{th:kahlernutbolt} Consider a K\"ahler metric with cohomogeneity-1 $SU(2)$ symmetry smooth at a nut or bolt I  at $\rho=0$ such that $a,b,c>0$ for $\rho>0$.  Then $a=b$ for all $\rho>0$, i.e. the K\"ahler metric has $U(1)\times SU(2)$ symmetry.
\end{prop}

\begin{proof}
Near a nut (\ref{nutexp}) immediately gives
\be
T = 1+O(\rho^2) \; ,  \label{TNBdy}
\ee
where smoothness dictates the higher order terms.
On the other hand, smoothness at a bolt I requires
\be
a= a_0+ O(\rho^2), \qquad b= a_0+ O(\rho^2), \qquad c= c_0 \rho (1+ O(\rho^2))  \; ,\label{boltIexp}
\ee
where $a_0, c_0 >0$ and the subleading terms are smooth in $\rho^2$ (so the space is diffeomorphic to $\mathbb{R}^2\times S^2$).  
Thus (\ref{boltIexp}) again gives (\ref{TNBdy}).
Therefore by Lemma \ref{lem:sym} we deduce that in both cases $T=1$ for $\rho>0$.
\end{proof}

Observe that only a nut and bolt I are compatible with an enhanced $U(1)\times SU(2)$ symmetry, therefore for bolt II one cannot obtain such  a result. It is worth noting that in the case of bolt I the K\"ahler form extends smoothly to the bolt and corresponds to the complex structure with respect to which the bolt is a complex submanifold. In contrast, a bolt II does not correspond to a complex submanifold. This provides an another distinguishing property between the two types of bolt and we will therefore sometimes refer to bolt I as a complex bolt.

In terms of Euler coordinates $(\theta, \phi, \psi)$ on $S^3$ (see Appendix \ref{app:su2}), the metric near a bolt I looks like 
\be
h \sim \td \rho^2+ c_0^2 \rho^2 (\td \psi+ \cos \theta \td \phi)^2+ a_0^2(\td \theta^2+ \sin^2\theta \td \phi^2) \; .
\ee
Therefore, the absence of the conical singularity at $\rho=0$ requires $\psi$ to be identified with period $2\pi /c_0$. Then the space near the bolt is diffeomorphic to $\mathbb{R}^2$ with coordinates $(\rho, \psi)$ fibered over an $S^2$ bolt with coordinates $(\theta, \phi)$. Regularity of this $\mathbb{R}^2$-bundle over the $S^2$  requires $\psi$ to be identified with period $4\pi/p$ where $p \in  \mathbb{N}$. Combining these regularity conditions we deduce that
\be
2 c_0 = p \in \mathbb{N}\;   \label{BIsmooth}
\ee
is required by smoothness near a bolt I.

It will be useful to consider a K\"ahler metric with $U(1)\times SU(2)$ symmetry  in the chart (\ref{EHgauge}). It is clear that it is a smooth Riemannian metric for $r>0$ provided $V(r)>0$ is smooth. The metric has potential singularities if $r=0$ or $V(r)=0$.  The former corresponds to a nut and smoothness of the metric requires $V(r)>0$ for all $r> 0$, $V(0)=1$ and $V(r)$ is a smooth function of $r^2$. The latter corresponds to a bolt (bolt I in the above) if $V(r)>0$ for $r>r_0>0$ where $V(r_0)=0$ and smoothness requires $V'(r_0)>0$ with $V(r)$ a smooth function of $r-r_0$; furthermore absence of the conical singularity in the $(r, \psi)$ part of the metric  implies $\psi$ must be identified with period $(8\pi)/(r_0 V'(r_0))$ and regularity of the resulting $\mathbb{R}^2$-bundle over the bolt implies
\be
\frac{1}{2}r_0 V'(r_0)=p \in \mathbb{N}  \label{EHbolt} \;. 
\ee
This latter condition is simply (\ref{BIsmooth}) written in the chart (\ref{EHgauge}).

It would be interesting to perform a systematic global analysis of K\"ahler metrics with $SU(2)$ symmetry. In the negative Einstein case it was found there are two families of complete metrics both containing bolts: the $U(1)\times SU(2)$ symmetric solution (\ref{EH}) (also see Appendix \ref{app:bolts}) and a triaxial diagonal metric~\cite{Dancer:1993aw} (these possess  a bolt I and bolt II respectively).   It would be interesting to perform a similar analysis for K\"ahler metrics satisfying other curvature conditions, such as extremal K\"ahler metrics or metrics satisfying the supersymmetric condition (\ref{PDE}). We will not pursue this here, although for some explicit examples see Appendix \ref{app:bolts}.

\section{Supersymmetric solutions with $SU(2)$ symmetry}
\label{sec:susysu2solns}

\subsection{General local solution}
\label{sec:local}

In this section we will construct the general timelike supersymmetric  solution with $SU(2)$ symmetry.  We will take the K\"ahler base metric  to be the general cohomogeneity-1 K\"ahler metric with $SU(2)$ symmetry as given in Theorem \ref{th:kahler}. The 1-form $\omega$ takes the $SU(2)$-invariant form (\ref{su2omega}).  

Firstly, the function $f$ is determined by the scalar curvature of the K\"ahler base (\ref{fGp})  which for our class of bases is given by (\ref{Rscalar}).
Next, we determine the 1-form $\omega$ using (\ref{dom}) which requires $G^\pm$. For this we need the Ricci form for the K\"ahler base which is given by the potential (\ref{Rform}).
Thus, $G^+$ is determined using (\ref{fGp}).  For $G^-$ we write (\ref{Gm}), where without loss of generality we may take the $\Omega^i$ to be given by (\ref{OmASD}). Thus by $SU(2)$-symmetry the $\lambda_i$ must be functions only of $\rho$ (since our $\Omega^i$ are $SU(2)$-invariant). Then  we find the integrability condition for (\ref{dom}) is equivalent to the pair of ODEs for $\lambda_i$, for $i=1,2$,
\be
\lambda_i'= \left(\frac{a^2+b^2-c^2-2 ab c'}{2ab c} \right)\lambda_i  \; ,\label{lambdaip} 
\ee
together with a complicated 5th order ODE for $c$ obtained using (\ref{ap}), (\ref{bp}) to eliminate all derivatives of $a,b$ (we do not give this as we will not need it). Note that the ODE for $c$ is in fact equivalent to $\Xi_a=0$ so the general PDE for the K\"ahler base (\ref{PDE})  is satisfied in a more restricted form. Then, upon use of these integrability conditions, we can solve (\ref{dom}) and find the general solution,
\be
\omega_1= -\frac{\ell^3}{48} c  b\lambda_1 , \qquad \omega_2 = -\frac{\ell^3}{48} a c \lambda_2   \label{om12}
\ee
and a complicated expression for $\omega_3$  that is 4th order in $c$ (again we will not need this). 

This completes the derivation of the general local solution with $SU(2)$ symmetry. Note that it depends on five functions $a,b,c, \lambda_1, \lambda_2$ subject to the first order ODEs (\ref{ap}), (\ref{bp}), (\ref{lambdaip}) and a complicated 5th order ODE for $c$.  In fact,  for solutions with an enhanced symmetry $U(1) \times SU(2)$ this system simplifies dramatically. We turn to this next.

We will now determine the general supersymmetric solution with $U(1)\times SU(2)$ symmetry.
This can be easily deduced from the general solution with $SU(2)$ symmetry given above. In order to obtain a K\"ahler base with this symmetry we must set $a^2=b^2$. To obtain $\omega$ of this symmetry we must set and $\omega_1=\omega_2=0$, which using (\ref{om12}) implies $\lambda_1=\lambda_2=0$. However, for convenience we will keep $\omega$ a general $SU(2)$-invariant form.

 In fact, we find it convenient to use the coordinate $r(\rho)$ in terms of which the K\"ahler metric is (\ref{EHgauge}).  In this coordinate system the Ricci form is given by the potential
\be
P= -\frac{1}{4}( rV'+ 4 V-4)\;  \sigma_3 \; , \label{RformV}
\ee
and the scalar curvature by
\be
R =- \frac{ 8(V-1)+7 r V' +r^2 V''}{ r^2} \; .  \label{RscalarV}
\ee 
One can then repeat the steps in the general case. This gives 
\be
f^{-1} = \frac{\ell^2( 8(V-1)+7 r V' +r^2 V'')}{24 r^2} \; .   \label{fsol}
\ee
The integrability condition for (\ref{dom}) is now a 5th order ODE for $V(r)$,\footnote{Curiously, (\ref{ODE}) happens to be a total derivative so can be integrated once.  Unfortunately this does not seem to help our later analysis.}
\bea
&&3 r^4 V V^{(5)}+ 6 r^4 V' V^{(4)}+ 30 r^3 V V^{(4)}+44 r^3 V' V^{(3)}+ r^2(47 V-32) V^{(3)} \nonumber \\ &&+ 8 r^3(V'')^2- 3 r( 13 V+32) V''+26 r^2 V' V''- 34 r V'^2+3 V'(13 V +32)=0  \label{ODE}
\eea
and  for $i=1,2$
\be
\lambda_i'=- \lambda_i \left( \frac{V'}{2V}+ \frac{2(V-1)}{rV} \right) \; .  \label{lambdaeq}
\ee
These correspond to the 5th order ODE for $c(\rho)$  and (\ref{lambdaip}) in the general solution, respectively. 
 Then, the 1-form $\omega$ is completely determined to be
\bea
\omega_3 &=& \frac{l^3}{2304 r^2} \left[-3 r^4 V''^2-64 r^2 V''+23 r^2 V'^2+V' \left(6 r^4 V^{(3)}+28 r^3 V''-320 r\right) \right. \nonumber  \\  && \left. +  2 V \left(3 r^4 V^{(4)}+30 r^3 V^{(3)}+77 r^2 V''+115 r V'-192\right)+192 V^2+192\right]  \; , \\
\omega_i &=& - \frac{\ell^3}{192} r^2 \sqrt{V} \lambda_i \; , \qquad i=1,2 \; .  \label{om12sym}
\eea
This completes the derivation of the solution in this case.

To summarise, the general local form of a supersymmetric solution with a $U(1)\times SU(2)$ symmetric base is determined by a single function $V(r)$ which obeys the 5th order ODE (\ref{ODE}). The rest of the data is uniquely fixed in terms of $V(r)$ as shown above. To obtain the general solution with $U(1)\times SU(2)$ one further requires $\omega$ to possess this symmetry which implies $\lambda_1=\lambda_2=0$.

\subsection{Examples: black holes and solitons}

In this section we record all the known supersymmetric solutions with $SU(2)$ symmetry as above. In fact they all have K\"ahler bases with $U(1)\times SU(2)$ symmetry and unless otherwise stated the full solution also has this enhanced symmetry (i.e. $\lambda_1=\lambda_2=0$). The solutions fall into two classes: black holes and solitons.

The Gutowski-Reall black hole is simply determined by~\cite{Gutowski:2004ez}
\be
V= 4\alpha^2 + \frac{r^2}{\ell^2} \; ,\label{GRBH}
\ee
where $\alpha > 1/2$ is a constant, which gives
\be
f =\frac{3 r^2}{ \ell^2 (4\alpha^2-1)+ 3 r^2}, \qquad \omega_3=\frac{(4\alpha^2-1)^2 \ell^3}{12r^2} + \tfrac{1}{2}(4\alpha^2-1)\ell+  \frac{r^2}{2\ell}   \; .
\ee
For $\alpha=1/2$ this solution reduces to global AdS$_5$ given by (\ref{VAdS}), (\ref{fomAdS}).  On the other hand, the near-horizon geometry of the Gutowski-Reall black hole (also a solution) is given simply by
\be
V= 4\alpha^2, \qquad f= \frac{3 r^2}{ \ell^2 (4\alpha^2-1)}, \qquad \omega_3=\frac{(4\alpha^2-1)^2 \ell^3}{12r^2}  \; , \label{GRNH}
\ee
where $\alpha>1/2$. Both of these represent spacetimes with smooth horizons at $r=0$ in line with the fact that $f=O(r^2)$ as $r\to 0$. 

Recently, another family of black hole solutions in this symmetry class have been constructed numerically, which are asymptotically locally AdS with conformal boundary that is spatially a homogeneously squashed $S^3$~\cite{Blazquez-Salcedo:2017kig, Blazquez-Salcedo:2017ghg, Cassani:2018mlh, Bombini:2019jhp}. We will comment further on these solutions at the end of Section \ref{sec:BH}.

We now turn to soliton solutions. In particular,  if $f>0$ globally, then the K\"ahler base must be a smooth K\"ahler surface and $\omega$ is a smooth 1-form on this surface. The corresponding supersymmetric solution is  then a strictly stationary soliton spacetime.  It is possible that there are soliton spacetimes with an ergosurface $f=0$, although we are not aware of any in this symmetry class. For simplicity we will only consider strictly timelike supersymmetric solitons. For these solutions, note that the Gram matrix $G$ of the Killing vectors $(V, R_i)$ with respect to the spacetime metric, which satisfies
\be
\det G = - f^{-1} a^2 b^2 c^2 \; ,
\ee
is invertible if and only if $abc\neq 0$. Therefore, one can say that the spacetime has a nut or bolt if and only if the K\"aher base has a nut or bolt as in Lemma \ref{lem:BC}. By applying Proposition \ref{th:kahlernutbolt} we may deduce the following general result for such solutions.

\begin{prop}
\label{prop:solitonsym}
Any strictly timelike supersymmetric soliton solution with $SU(2)$ symmetry containing a nut or a bolt I must have a K\"ahler base with $U(1)\times SU(2)$ symmetry. 
\end{prop}

The only  known nontrivial soliton  with a nut  has been constructed numerically and is asymptotically locally AdS$_5$ with squashed $S^3$ spatial  boundary geometry~\cite{Cassani:2014zwa} (this is given by the branch B solutions defined by (\ref{branchB}) and (\ref{Vseries}) discussed as the end of Section \ref{sec:BH}). It is natural to wonder whether there are soliton solutions with a bolt I, which by the above result must have enhanced $U(1)\times SU(2)$ symmetry. We may construct such solutions as follows.  

First, consider the most general  K\"ahler-Einstein metric with $U(1)\times SU(2)$ symmetry, which  is given by
\be
V= 1+ \frac{r^2}{\ell^2} +\frac{c_4}{r^4}  \; ,\label{EH}
\ee
where $c_4$ is a constant. This is a generalisation of the Eguchi-Hanson metric (obtained by setting $\ell \to \infty$), which  reduces so the Bergmann metric for $c_4=0$. It is easily checked that it is a solution to (\ref{ODE}), or simply recall that as noted earlier any Einstein base gives a supersymmetric solution.  The rest of the data $(f, \omega)$ is identical to that for AdS$_5$ (\ref{fomAdS}). While this gives a complete K\"ahler metric  with a bolt I (for $c_4<0$), unfortunately, the resulting supersymmetric solution is either singular or has CTC  since $\omega_3$ can never vanish at the bolt (see the Appendix \ref{app:bolts}).

A more general class of solutions can be constructed as follows. Inspired by the Gutowski-Reall and K\"ahler-Einstein bases, consider the ansatz
\be
V= c_0+\frac{r^2}{\ell^2} + \frac{c_2}{r^2}+\frac{c_4}{r^4} \label{GREH}
\ee
where $c_0, c_2, c_4$ are constants. We find that this is a solution to (\ref{ODE}) iff \footnote{In fact, this solution was also noticed in~\cite{Cassani:2015upa}.}
\be
c_2^2= 3 (c_0-1) c_4  \; .
\ee
The special case $c_4=0$ gives the Gutowski-Reall solution (if $c_0>1$), whereas special case $c_0=1$ gives the K\"ahler-Einstein solution. This class also contains complete K\"ahler metrics with a bolt I (see Appendix \ref{app:bolts}). Interestingly, in this case we find  that these bases do give smooth soliton solutions.  Recall, smoothness of the K\"ahler base requires the existence of a constant $r_0>0$ such that $V(r_0)=0$ and $V(r)>0$ for $r>r_0$ and (\ref{EHbolt}) for some $p \in \mathbb{N}$. Furthermore, smoothness of the spacetime requires $f>0$ for $r\geq r_0$ and $\omega_3(r_0)=0$.  We find that solutions do exist but only for $p\geq 3$, see Appendix \ref{app:bolts} for details.  In particular,  we obtain asymptotically locally AdS$_5$ solitons with $S^3/\mathbb{Z}_p$ spatial boundary and a bolt at $r=r_0$.  It turns out that these solutions were previously found as limits of non-supersymmetric solutions~\cite[section 3.4]{Cvetic:2005zi} (case B), however it is not apparent from their analysis that the K\"aher base has $U(1)\times SU(2)$ symmetry (i.e. that supersymmetry commutes with the spacetime $U(1)\times SU(2)$ symmetry)\footnote{Indeed, there is an example of a supersymmetric soliton with $U(1)\times SU(2)$ symmetry~\cite[section 3.3]{Cvetic:2005zi} (case A), for which the supersymmetric Killing field $V$ does not commute with the $SU(2)$ symmetry~\cite{Cassani:2015upa}. As a result the K\"ahler base of this solution only has $U(1)^2$ symmetry.}, and furthermore a detailed analysis of the allowed values of $p\in \mathbb{N}$ was not performed.

Of course, it would be interesting to determine all soliton solutions with a nut or bolt I.  This would require classifying all solutions to (\ref{ODE}) that are compatible with such boundary conditions, which we will not pursue here.  It would also be interesting to investigate the existence of solitons with a bolt II; such solutions may have only generic $SU(2)$ symmetry and we are not aware of any examples in this class.

Finally, it is worth noting that although the only known explicit solutions have a K\"ahler base with $U(1)\times SU(2)$ symmetry, one can easily construct supersymmetric solutions with exactly $SU(2)$ symmetry from such K\"ahler bases by taking $\omega$ to be a general $SU(2)$ invariant 1-form. This amounts to taking $\lambda_1, \lambda_2 \neq 0$ above and solving (\ref{lambdaeq}). In particular, for the Gutowski-Reall  base one finds the following deformation of the Gutowski-Reall black hole,
\be
\omega_i = - \frac{c_i\ell^3 }{192} \left(\frac{r^2}{ 4\alpha^2 + \frac{r^2}{\ell^2} }\right)^{\frac{1}{4\alpha^2}} \; , \label{GRdef}
\ee
for $i=1,2$ where $c_i$ are constants, with the rest of the data being identical.  Even though $\omega_{1,2}\to 0$ as $r \to 0$, the results of Section \ref{sec:BH} show that presence of such terms is incompatible with a regular horizon at $r=0$.  However,   for $\alpha=1/2$ these give smooth deformations of AdS$_5$ studied in~\cite{Gauntlett:2004cm}.   On the other hand, for any soliton with a bolt I at $r=r_0$ as in Proposition \ref{prop:solitonsym}, one finds that as $r \to r_0$,
\be
\omega_i = c_i (r-r_0)^{1/p} (1+O(r-r_0))  \; ,
\ee
for $i=1,2$, where we have used (\ref{EHbolt}).\footnote{This easily follows from the general formula $\omega_i'/\omega_i=2/(r V)$ which is a consequence of (\ref{lambdaeq}) and (\ref{om12sym}).} Therefore, for the asymptotically AdS$_5/\mathbb{Z}_p$ solitons discussed above which must have $p\geq 3$, we deduce that these do not give smooth deformations.

\section{Black hole solutions with $SU(2)$ symmetry}
\label{sec:BH}

We now consider the constraints imposed by having a regular event horizon. First, recall that any Killing vector field of a spacetime $(M,g)$ must be tangent to the event horizon.  This implies that any Killing field restricted to the event horizon is null (and tangent to the generators of the horizon) or spacelike. Therefore, since for a supersymmetric solution $|V|^2=-f^2 \leq 0$ this means that the supersymmetric Killing field must be null on the horizon, i.e. the event horizon is a Killing horizon of $V$. Furthermore, since $| V |^2$ is at a global maximum at the horizon we must also have that $\td | V |^2=0$ on the horizon, i.e. it is an extremal horizon.

Next, we assume that $(M,g)$ has a $G$-symmetry, where $G$ is $SU(2)$, as defined above Lemma \ref{lem:Gsym}. Therefore, the Killing fields $R_i$ that generate $SU(2)$ must be tangent to the horizon. Now, since we also assume the orbits are spacelike, we can always choose a cross-section of the horizon tangent to $R_i$, i.e., a $G$-invariant cross-section. Furthermore, since we assume $G$ has 3d orbits the cross-sections are homogeneous spaces locally isometric to $SU(2)$. Thus,  the only possible horizon topologies are $S^3$ and lens spaces.

\subsection{Near-horizon analysis}

We will now analyse constraints on the geometry arising from a regular event horizon. In particular we assume the event horizon is a smooth (or analytic) null hypersurface with a smooth cross-section. We can write the metric near the horizon in Gaussian null coordinates (see e.g.~\cite{Kunduri:2013gce}), which as argued above can be adapted to an $SU(2)$-invariant cross-section. Thus, the general $SU(2)$-invariant metric near a horizon takes the form,
\be
g= -\lambda^2 \Delta^2 \td v^2+ 2 \td v \td \lambda+ 2\lambda h_i \hat{\sigma}_i \td v + \gamma_{ij} \hat{\sigma}_i \hat{\sigma}_j  \; ,   \label{GNC}
\ee
where $\partial_\lambda$ are null geodesics transverse to the horizon synchronised so the horizon is at $\lambda=0$, the supersymmetric Killing field $V=\partial_v$, the $\hat{\sigma}_i$ are right-invariant 1-forms that satisfy (\ref{MCsu2}), and  the data $\Delta, h_i, \gamma_{ij}$ are  smooth (or analytic) functions of only $\lambda$ and $\gamma_{ij}$ is a positive-definite matrix.  Note that outside the horizon, $\lambda >0$, we assume $\Delta$ is non-zero as required for solutions in the timelike class. The near-horizon geometry is given by the scaling $(v, \lambda)\to (v/\epsilon, \epsilon \lambda)$ and the limit $\epsilon\to 0$, in which case the metric still takes the form (\ref{GNC}) where the data $\Delta, h_i, \gamma_{ij}$ are replaced by their values at $\lambda=0$ which we denote by $\mathring{\Delta}, \mathring{h}_i, \mathring{\gamma}_{ij}$. In particular, observe that  the near-horizon data $\mathring{\Delta}, \mathring{h}_i, \mathring{\gamma}_{ij}$ are all constants.

In fact  supersymmetric near-horizon geometries with compact horizon cross-sections have been completely classified in this theory~\cite{Gutowski:2004ez, Kunduri:2006uh,Grover:2013hja}.  It turns out that the only solution with $SU(2)$-symmetry is given by the original Gutowski-Reall near-horizon geometry, which can be written in the form~\cite{Gutowski:2004ez},
\bea
&&\mathring{h}_1=\mathring{h}_2=0, \qquad \mathring{h}_3= -\frac{3 \mathring{\Delta}}{\ell (\mathring{\Delta}^2- 3\ell^{-2})} \nonumber \\
&&\mathring{\gamma}_{11}=\mathring{\gamma}_{22}=\frac{1}{\mathring{\Delta}^2- 3\ell^{-2}}, \qquad  \mathring{\gamma}_{33}= \frac{\mathring{\Delta}^2}{(\mathring{\Delta}^2- 3\ell^{-2})^2},  \label{GRNH}
\eea
where $\mathring{\gamma}_{ij}$ is diagonal and $\mathring{\Delta}>\sqrt{3}/\ell$. We emphasise that the near-horizon geometry has enhanced $U(1)\times SU(2)$ symmetry. This will be important below.

We now compare certain invariants to those for a general timelike supersymmetric solution. Comparison of  the scalar $|V|^2$ implies 
\be
f=\lambda \Delta  \;  ,   \label{fNH}
\ee
where we have chosen a sign so that $\Delta\geq 0$.  
Furthermore, the metric $h$ on the orthogonal base (which is invariantly defined wherever $V$ is timelike) can be extracted  to give
\bea
h &=& \frac{1}{\Delta \lambda} ( \td \lambda+ \lambda h_i \hat{\sigma}_i)^2 + \lambda\Delta \gamma_{ij} \hat{\sigma}_i \hat{\sigma}_j \nonumber \\
&=&\left( \frac{ \Delta }{\Delta^2+ h_ih^i}\right) \frac{\td \lambda^2}{\lambda}+ \lambda\Delta q_{ij} \left(\hat{\sigma}_i +\frac{k^i \td \lambda}{\lambda} \right) \left(\hat{\sigma}_j +\frac{k^j \td \lambda}{\lambda}  \right) \; , \label{hGNC}
\eea
for $\lambda>0$,  where for convenience we introduce $h^i:=\gamma^{ij}h_j$, 
\be
q_{ij}:= \gamma_{ij}+ \frac{h_i h_j}{\Delta^2}, \qquad  k^i := \frac{h^i}{\Delta^2+ h^i h_i}.
\ee
Note these obey $q_{ij} k^j= h_i/\Delta^2$.  Finally, we can also extract the 1-form $\omega$, which gives
\be
\omega= -\frac{h_i \hat{\sigma}_i}{\lambda \Delta^2} \; , \label{omGNC}
\ee
where without loss of generality we have fixed the $\lambda$ component to zero since the $SU(2)$ symmetry implies it is pure gauge (here we using the gauge freedom in the definition of $\omega$).  We may now establish the following converse to Lemma \ref{lem:Gsym} for solutions with horizons which will be useful later.

\begin{lemma} \label{lem:Hsym}
Consider a  supersymmetric  solution with $SU(2)$ symmetry that is timelike outside a smooth horizon. If the K\"ahler base has $U(1)\times SU(2)$ symmetry then the spacetime metric also has $U(1)\times SU(2)$ symmetry, i.e., one can take $\omega$ to be $U(1)\times SU(2)$ invariant. 
\end{lemma}

\begin{proof}
Suppose the K\"ahler base metric $h$ has $U(1)\times SU(2)$ symmetry, i.e. assume (\ref{hGNC}) satisfies $\mathcal{L}_{\hat{L}_3} h=0$ where $\hat{L}_i$ are the right-invariant vectors dual to $\hat{\sigma}_i$. Then it easily follows that $h_1=h_2=0$ and that $q_{ij}$ is diagonal with  $q_{11}=q_{22}$.  By the definition of $q_{ij}$ this implies $\gamma_{ij}$ is diagonal with  $\gamma_{11}=\gamma_{22}$. Thus the spacetime metric (\ref{GNC}) has $U(1)\times SU(2)$ symmetry as claimed. In particular, the 1-form $\omega$ in the gauge (\ref{omGNC})  is $U(1)\times SU(2)$ invariant.
\end{proof}

We now wish to bring the metric (\ref{hGNC}) into the standard form (\ref{su2metric}). To this end, let us define the 1-forms
\be
{\sigma}_i := A_{ij}(\lambda)( \hat{\sigma}_j + B_j(\lambda) \td \lambda) \; ,  \label{sigmatrans}
\ee
where $A_{ij}$ are the components of an invertible matrix.  Then we find
\be
\td {\sigma}_i = - \tfrac{1}{2} \epsilon_{pjk} A_{ip} (A^{-1})_{jl} (A^{-1})_{km} {\sigma}_l \wedge {\sigma}_m  + \td \lambda\wedge {\sigma}_m \left( A_{ij}' (A^{-1})_{jm} + A_{ip} (A^{-1})_{km} B_j \epsilon_{pjk} \right)  \; ,
\ee
and  requiring that the terms proportional to $\td \lambda$ vanish is equivalent to
\be
A'_{ij}= A_{ik} C_{kj}, \qquad \text{where} \qquad C_{ij} := \epsilon_{ijk} B_k  \; . \label{AB}
\ee
Now $C_{ij}$ is an antisymmetric matrix, so this equation can always be solved locally for some orthogonal matrix $A_{ij}$.  In this case
\be
 \epsilon_{pjk} A_{ip} (A^{-1})_{jl} (A^{-1})_{km} =  \epsilon_{pjk} (A^{-1})_{pi} (A^{-1})_{jl} (A^{-1})_{km} = \det (A^{-1}) \epsilon_{ilm} = \pm \epsilon_{ilm}  \; .
 \ee
 In particular, if we take $A\in SO(3)$ we deduce that the 1-forms $\sigma_i$ defined by (\ref{sigmatrans}) satisfy the Maurer-Cartan equations for right-invariant 1-forms (\ref{MCsu2}) for $SU(2)$.
Given this, we can  introduce a new set of Euler coordinates on $SU(2)$ such that  ${\sigma}_i$ are right-invariant 1-forms.
 
 We may use this result  to transform (\ref{hGNC}) into the standard form (\ref{su2metric}) as follows.  Perform a coordinate transformation such that $(\lambda, \hat{\sigma}_i) \to (\rho, {\sigma}_i)$, where $\rho= \rho(\lambda)$ and (\ref{sigmatrans}) are defined by 
 \bea
 \left(\frac{\td \rho}{\td \lambda} \right)^2 &=&   \left( \frac{ \Delta }{\Delta^2+ h_ih^i}\right) \frac{1}{\lambda}  \; ,\label{rholambda} \\
 B_i &=& \frac{k^i}{ \lambda}\; ,  \label{B}
 \eea
 where $A_{ij}\in SO(3)$ obeys (\ref{AB}). Then  we obtain (\ref{su2metric}) where 
 \be
 {h}_{ij}= \lambda \Delta A_{ik} q_{kl} A^T_{lj} \; .  \label{hath}
 \ee
 The 1-form $\omega$ in the new coordinates is (gauge equivalent) to (\ref{su2omega}) where 
 \be
 \omega_i=  -\frac{A_{ij} h_j}{\lambda \Delta^2} \; .  \label{omhat}
 \ee
We will now use this to derive the boundary conditions near a horizon.

\begin{prop} \label{prop:NH}
Consider a  timelike supersymmetric and $SU(2)$-symmetric solution to $D=5$ minimal gauged supergravity containing a smooth  horizon with $\mathring{\Delta}>0$ and an $U(1)\times SU(2)$-invariant near-horizon geometry. The horizon corresponds to a conical singularity in the K\"ahler base metric $h$, i.e., in the metric (\ref{su2metric}) the horizon can be taken to be  at $\rho=0$ such that
\be
{h}_{ij}=  \rho^2 \mathring{h}_{ij}+O(\rho^4)   
\label{conical}
\ee
as $\rho \to 0$, where $\mathring{h}_{ij}$ is a positive diagonal matrix with $\mathring{h}_{11}= \mathring{h}_{22}$ fixed by the near-horizon geometry.
Furthermore, the 1-form $\omega$ is (\ref{su2omega}) where
\be
\omega_i = \frac{ \mathring{\omega}_i }{\rho^2} + O(1)  \label{omNH}
\ee
and  $\mathring{\omega}_i=\mathring{\omega}_3 \delta_{i3}$ is a constant determined by the near-horizon geometry. Finally, 
\be
f = \rho^2 \mathring{f} + O(\rho^4),
\ee
where $\mathring{f}$ is a positive constant determined by the near-horizon geometry and the subleading terms are smooth functions of $\rho^2$. 
\end{prop}
 
 \begin{proof}
 As mentioned above a  smooth (or analytic) horizon requires that the data $\Delta, h_i, \gamma_{ij}$  be smooth (or analytic) functions of $\lambda$ at $\lambda=0$. Then, setting the horizon to be at $\rho=0$,  the ODE (\ref{rholambda})  implies that $\rho^2$ is a smooth (or analytic) function of $\lambda$ at $\lambda=0$ since by assumption $\Delta> 0$ at $\lambda=0$.   Indeed, explicitly solving (\ref{rholambda}) near $\lambda=0$ gives the leading order behaviour
 \be
 \rho^2= \left( \frac{4\mathring{\Delta}}{\mathring{\Delta}^2+ \mathring{h}_i \mathring{h}^i} \right) \lambda  \left( 1+ O(\lambda) \right) \; . \label{lambdarho}
 \ee
Inverting, we also deduce that $\lambda$ is a smooth (or analytic) function of $\rho^2$ at $\rho=0$.
 Next we need to determine $A_{ij}$ from (\ref{AB}) where $B_i$ is given by (\ref{B}). Clearly, if $\mathring{k}^i=0$ then $B_i$ is smooth at the horizon so there is a  solution $A_{ij}$ that is smooth at the horizon.   However, if $\mathring{k}^i \neq 0$ (as is the case for the Gutowski-Reall near-horizon geometry) then $B_i$ is singular at the horizon, which in turn implies that $A_{ij}$ is not $C^1$ at $\lambda=0$ (indeed, it is not even $C^0$).   Nevertheless, since $A_{ij}$ is an orthogonal matrix, its entries are bounded so $A_{ij}= O(1)$ as $\lambda \to 0$, which will be important in bounding the subleading terms in what follows.
 
 In order to isolate the singular behaviour in $A$  define the parameter $\tau:= -\log \lambda$ so (\ref{AB}) becomes $\td A/ \td \tau= -A K$ where $K_{ij}=\epsilon_{ijk} k^k$ is a smooth function of $\lambda= e^{-\tau}$. Therefore a unique solution exists on $[\tau_0, \infty)$ with a given initial condition at $A(\tau_0)$ for some $\tau_0$ (determined by the range of the coordinate $\lambda$). In order to examine the behaviour of $A$ near the horizon we need to determine its asymptotics as $\tau\to \infty$.  For large $\tau$  we have $K=K_0+ O(e^{-\tau})$ where $K_{0ij}= \epsilon_{ijk} \mathring{k}^k$. It follows that
 \be
 \frac{\td }{\td \tau} ( A e^{\tau K_0} ) = -A (K-K_0) e^{\tau K_0} = O(e^{-\tau}) \; ,
 \ee
 where the final estimate follows from the fact that $A, e^{\tau K_0}$ are orthogonal matrices (so have bounded components) and the aforementioned asymptotics of $K$. Integrating, it follows that\footnote{Integrating over $[\tau, \tau_*]$, for $\tau>\tau_0$, gives $A(\tau_*) e^{\tau_* K_0}- A(\tau) e^{\tau K_0}= -\int_{\tau}^{\tau_*} A (K-K_0) e^{\tau K_0} \td \tau$  and since the integrand is $O(e^{-\tau})$ the limit $\tau_*\to \infty$ exists, so $A_0-A(\tau) e^{\tau K_0}=  -\int_{\tau}^{\infty} A (K-K_0) e^{\tau K_0} \td \tau$ where $A_0 := \lim_{\tau_*\to \infty} A(\tau_*) e^{\tau_* K_0}$. Then, use again that the integrand is $O(e^{-\tau})$ to deduce $A_0-A(\tau) e^{\tau K_0}= O(e^{-\tau})$.}
 \be
 A= A_0 e^{-\tau K_0} + O(e^{-\tau})  = A_0 e^{K_0 \log \lambda} + O(\lambda)  \; ,
 \ee
 where $A_0$ is a constant orthogonal matrix. Clearly we are free to left multiply $A$ by a constant orthogonal matrix and we will use this to set $A_0=I$.
  
 We now wish to expand the metric data near the horizon.  For this we will need to use the assumptions that the near-horizon geometry has $\mathring{\Delta}>0$ and enhanced $U(1)\times SU(2)$ symmetry.  Crucially, this means $\mathring{\gamma}_{ij}$ is diagonal with $\mathring{\gamma}_{11}=\mathring{\gamma}_{22}$ and $\mathring{h}_1=\mathring{h}_2=0$ and hence $\mathring{q}_{ij}$ is diagonal with $\mathring{q}_{11}=\mathring{q}_{22}$. 
 Now, expanding (\ref{hath}) we find
 \be
 {h}_{ij} = \lambda \mathring{\Delta} \mathring{q}_{ij}  + O(\lambda^2)  \; ,
 \ee
 where we have used the fact that $e^{-\tau K_0} \mathring{q} e^{\tau K_0}=\mathring{q}$ and it should be noted that the subleading $O(\lambda^2)$ terms are not necessarily smooth at $\lambda=0$ in general (due to the factors of $e^{\tau K_0}$).  Thus, using (\ref{lambdarho}) we obtain (\ref{conical}) with
 \be
 \mathring{h}_{ij}=\frac{1}{4} (\mathring{\Delta}^2+ \mathring{h}_i \mathring{h}^i) \mathring{q}_{ij}  \; .
 \ee
 Similarly, expanding (\ref{omhat}) we obtain
 \be
 {\omega}_i = -  \frac{\mathring{h}_i}{\lambda \mathring{\Delta}^2} + O(1) \; ,
 \ee
 where we have used that $(e^{-\tau K_0})_{ij} \mathring{h}_j=\mathring{h}_i$, and again the $O(1)$ terms are not generally smooth at the horizon, 
so we find  (\ref{omNH}) where
 \be
 \mathring{\omega}_i =  - \frac{4 \mathring{h}_i} {\mathring{\Delta}(\mathring{\Delta}^2+ \mathring{h}_i \mathring{h}^i)}  \; .
 \ee
Finally, (\ref{fNH}) implies
\be
f=\frac{1}{4} (\mathring{\Delta}^2+ \mathring{h}_i \mathring{h}^i) \rho^2  + O(\rho^4) \; ,
\ee
where now the higher order terms are smooth functions of $\rho^2$. This completes the proof.
 \end{proof}

Now, as mentioned above the explicit form of near-horizon geometry must be given by the Gutowski-Reall near-horizon geometry (\ref{GRNH}), which in particular satisfies the conditions in the above proposition. The constants appearing in the above can be computed explicitly from  (\ref{GRNH}) which gives
\bea
&&\mathring{h}_{11}= \mathring{h}_{22}= \alpha^2, \qquad \mathring{h}_{33}= 4 \alpha^4  \; ,\\
&& \mathring{f}= \frac{12 \alpha^2}{\ell^2 (4\alpha^2-1)}, \qquad \mathring{\omega}_3= \frac{\ell^3 (4\alpha^2-1)^2}{48 \alpha^2} \; ,
\eea
where we have introduced the constant
\be
\alpha^2:= \frac{1}{4}\frac{\mathring{\Delta}^2+ \frac{9}{\ell^2}}{\mathring{\Delta}^2 - \frac{3}{\ell^2}}  \; ,
\ee
which satisfies $\alpha>1/2$.

We will now combine our near-horizon analysis with the general constraints for a timelike supersymmetric solution with $SU(2)$ symmetry. In particular, this means the metric $h$ on the orthogonal base must be K\"ahler with cohomogeneity-1 $SU(2)$ symmetry.  The general form for such a K\"ahler metric is given in Theorem \ref{th:kahler}. However,  it is important to note that this assumes that one can use the rotation freedom $\sigma_i \to R_{ij} \sigma_j$ where $R \in SO(3)$ is constant, to diagonalise $h_{ij}(\rho)$ for some value of $\rho$ where $h_{ij}$ is invertible. Now, inspecting our proof of Proposition \ref{prop:NH}, reveals that only an $SO(2)$ subgroup which fixes $\sigma_3$ of this rotation freedom remains (that which preserves $\mathring{q}_{ij}$).   Fortunately, this is just enough for our purposes. Indeed, the proof of Theorem \ref{th:kahler} shows that without using this $SO(3)$ rotational freedom, the general K\"ahler metric is described in our frame (\ref{nonDframe}) by four functions $a, b, c, b_1$, so the only non-diagonal term in the metric is $2 a b_1 \sigma_1 \sigma_2$. Therefore,  we can always diagonalise this metric at a point using the $SO(2)$ rotations that fix $\sigma_3$.  We deduce that may indeed use the general form for the K\"ahler metric as given in Theorem \ref{th:kahler}, which in particular is diagonal for $\rho>0$.

 Therefore, comparing the near-horizon expansion given in Proposition \ref{prop:NH} to the K\"ahler metric in Theorem \ref{th:kahler}, implies that near a horizon $\rho=0$ the functions in the K\"ahler metric are given by
\be
a= \alpha \rho+ O(\rho^3),  \qquad b= \alpha \rho+ O(\rho^3), \qquad c= 2\alpha^2 \rho+ O(\rho^3) \; .  \label{abcNH}
\ee
These near-horizon expansions will be important in our subsequent analysis.
We emphasise that due to the singular behaviour of the coordinate change (\ref{sigmatrans}) defined by $A_{ij}$, the subleading terms are not necessarily smooth. The singular behaviour of $A_{ij}$ at the horizon prevents us from improving this result in the general case. However, if one has an $U(1) \times SU(2)$ symmetry one can obtain a stronger result which allows us to control the regularity of the subleading terms.

\begin{prop} \label{prop:NHu1su2}
Consider a timelike supersymmetric solution to $D=5$ minimal gauged supergravity with $U(1)\times SU(2)$ symmetry containing a smooth (analytic)    horizon with the Gutowski-Reall near-horizon geometry (\ref{GRNH}). The horizon corresponds  to a conical singularity in the K\"ahler base metric $h$, i.e., in the metric (\ref{hsu2u1}) the horizon can be taken to be  at $\rho=0$ such that
\be
a^2= \alpha^2 \rho^2 +O(\rho^4), \qquad c^2= 4\alpha^4 \rho^2 + O(\rho^4)  \; ,
\ee
as $\rho \to 0$ and the metric components $a^2, c^2$  are smooth (analytic) functions of $\rho^2$ at $\rho=0$.  
Furthermore, the 1-form $\omega$ takes the form (\ref{omsu2u1}) where
\be
 \omega_3= \frac{\mathring{\omega}_3}{\rho^2} + O(1)
\ee
and $\rho^2 \omega_3$ is a smooth (analytic) function of $\rho^2$ at $\rho= 0$. The function $f$ is as in Proposition \ref{prop:NH} and is a smooth (analytic) function of $\rho^2$.

 Conversely, any timelike supersymmetric solution $(f, \omega, h)$ of this form, has a smooth (analytic) horizon at $\rho=0$.
\end{prop}

\begin{proof}
We may specialise the proof of the previous proposition. The $U(1) \times SU(2)$ symmetry implies that we can write (\ref{GNC}) where $\gamma_{ij}$ must be diagonal with $\gamma_{11}=\gamma_{22}$ and $h_1=h_2=0$ for all $\lambda$. Therefore, the coordinate change defined by (\ref{AB}) simplifies: since $B^1=B^2=0$ the matrix $A_{ij}$ must be of the block diagonal form
\be
A= \begin{pmatrix}  P & 0 \\ 0 & 1 \end{pmatrix} \; ,
\ee
where $P \in SO(2)$. It then follows from (\ref{hath}), (\ref{omhat}) that for all $\lambda>0$,
\be
{h}_{ij} =\lambda \Delta q_{ij}, \qquad {\omega}_i = -\frac{h_i}{\lambda \Delta} \; ,  \label{hhatsym}
\ee
so the matrix $A$ has dropped out of these expressions.  The result now follows upon use of (\ref{lambdarho}) and that $\lambda$ is a smooth (analytic) function of $\rho^2$, which were shown in Proposition \ref{prop:NH}. The converse statement is a result of reversing the above steps.
\end{proof}

It is convenient to restate this result in the chart (\ref{EHgauge}).

\begin{cor} \label{cor:u1su2BH}Consider a timelike supersymmetric and $U(1)\times SU(2)$-invariant solution with a smooth (analytic)  horizon. The horizon corresponds to a conical singularity at $r=0$ in the K\"ahler base (\ref{EHgauge}), and $V(r)>0$,  $ r^{-2} f>0$, $r^2 \omega_3$ are all smooth (analytic) functions of $r^2$ at $r=0$. Conversely, any solution of this form has a smooth (analytic) horizon at $r=0$.
\end{cor}

This simply follows from the coordinate change $r^2= 4 a(\rho)^2$ together with the form of $a^2$ in Proposition \ref{prop:NHu1su2}, which  imply that $4 \alpha^2 \rho^2 = r^2 (1+ O(r^2))$ where the higher order terms are smooth (analytic) functions of $r^2$.

\subsection{Uniqueness theorem}

In this section we will prove our main result which is Theorem \ref{thm} stated in the Introduction. First we will prove the following symmetry enhancement result.

\begin{prop} \label{prop:symBH}
Any  timelike supersymmetric and $SU(2)$-invariant   solution containing a smooth horizon with compact cross-sections must have $U(1)\times SU(2)$ symmetry.
\end{prop}

\begin{proof} 
The near horizon behaviour derived in Proposition \ref{prop:NH} in particular implies (\ref{abcNH}) and hence the function (\ref{Tdef}) satisfies 
\be
T= 1+ O(\rho^2), 
\ee
near a horizon $\rho=0$.  For $\rho>0$ the solution is timelike so we may assume the functions $a, b, c>0$. Therefore, Lemma \ref{lem:sym} implies  $T=1$ for all $\rho>0$, i.e. $a=b$ for $\rho>0$.  Therefore the K\"ahler base must have $U(1)\times SU(2)$ symmetry.  Finally, by Lemma \ref{lem:Hsym} we deduce that the spacetime metric must also have $U(1)\times SU(2)$ symmetry.
\end{proof}

We are now ready to prove the black hole uniqueness theorem.  \\

\noindent {\it Proof of Theorem \ref{theorem}}. 
Proposition \ref{prop:symBH} shows that the solution must have $U(1)\times SU(2)$ symmetry.  In the coordinate system (\ref{EHgauge}) the local classification of solutions in this symmetry class in given in section \ref{sec:local}, in terms of a single function $V(r)$ which obeys a 5th order ODE (\ref{ODE}).  Furthemore, the conditions for such solutions to possess an analytic horizon are given in Proposition \ref{prop:NHu1su2} and Corollary \ref{cor:u1su2BH}. In particular, an analytic horizon corresponds to $r=0$ where $V(r)$ is a positive analytic function of $r^2$ and $V(0)>0$. 

 In fact, for the sake of generality we will assume $V(r)$ is an analytic function of $r$, so
\be
V(r)= \sum_{n=0}^\infty \frac{V_n r^n}{n!}  \; .  \label{Vseries}
\ee
We can derive constraints on the first few coefficients by expanding each term in (\ref{ODE}). We find that the vanishing of the $r^1$ term implies $V_1=0$ which then also implies the $r^0$ terms vanish. The $r^2$ terms then give
\be
V_3 \left( V_0 -\frac{32}{11} \right)=0  \label{r2}
\ee
and the $r^3$ terms give
\be
V_4 \left( V_0 -1 \right)=0  \; .  \label{r3}
\ee
In general, this will lead to different branches of solutions. However, for black hole boundary conditions we can derive more constraints.

First, by analyticity, the expansion for $f$ can be obtained from (\ref{fsol}), which gives
\be
f^{-1}= \frac{(V_0-1)\ell^2}{3 r^2} + O(1) \; .
\ee
Therefore, by Corollary \ref{cor:u1su2BH}, which requires $r^{-2}f$ is positive for $r\geq 0$, we deduce $V_0> 1$  (note this also ensures $f>0$  for small $r>0$).  Now, since an analytic horizon requires $V(r)$ to be an analytic function of $r^2$, in particular we must have $V_3=0$.  Therefore the constraints (\ref{r2}), (\ref{r3}) on the coefficients are satisfied provided $V_4=0$ in which case (\ref{ODE}) is satisfied up to order $r^3$.

We will now prove that $V_0>1$ and  $V_3=0$ implies $V_n=0$ for all $n\geq 4$. We proceed by induction and have already verified the base case $n=4$. Thus our induction hypothesis is
\be
V(r)= V_0+\frac{V_2 r^2}{2} + \frac{r^{n} V_{n}}{n!}+ O(r^{n+1}),
\ee
for some $n \geq 4$, where the higher order terms are analytic, and we wish to prove this implies $V_{n}=0$. Substituting this into (\ref{ODE}) we find
\be
\frac{[ 3V_0(n^2-16)+32 (V_0-1)] (n^2-4) V_{n} r^{n-1}}{(n-1)!} + O(r^n)=0  \; .
\ee
Since $V_0>1$ guarantees the factor in the square brackets is positive for all $n\geq 4$, this implies $V_{n}=0$ as required. Therefore, by induction this shows that $V_n=0$ for all $n\geq 4$ as claimed.

Thus we have  shown that the only analytic solution to (\ref{ODE}) with $V_0>1$ and $V_3=0$ is given by
\be
V(r)= V_0+ \frac{V_2}{2} r^2  \; .  \label{Vquad}
\ee
Using (\ref{fsol}) this fixes
\be
f= \frac{6 r^2}{\ell^2[ 2(V_0-1) + 3 r^2 V_2]} \; .
\ee
If $V_2=0$ this is the near-horizon geometry of the Gutowski-Reall black hole with $V_0=4\alpha^2$ (\ref{GRNH}).   If $V_2 \neq 0$, then since the invariant   $f$ must be smooth for $r>0$ we must have $V_2>0$. This corresponds to the Gutowski-Reall black hole  with $V_0=4\alpha^2$ and   $V_2=2/\ell^2$ (\ref{GRBH}).  $\Box$ \\

For this result we  needed to assume a stronger regularity property, namely that the metric at the horizon is analytic. 
It is worth emphasising that there are other solutions to (\ref{ODE}) that are analytic at $r=0$. In particular, the above proof shows two other branches of solutions defined by (recall in all cases $V_1=0$),
\bea
&&\text{Branch B}: \qquad V_0=1, \qquad V_3=0  \label{branchB} \\
&&\text{Branch C}: \qquad V_0=32/11, \qquad V_4=0.
\eea
It can be shown that branch B is uniquely determined in terms of $V_2$ and is an even function of $r$. It has been constructed numerically and corresponds to a smooth soliton with a nut at $r=0$ which is asymptotically locally AdS$_5$ with a squashed $S^3$ at infinity~\cite{Cassani:2014zwa}.  

On the other hand, branch C is determined uniquely in terms of $V_2, V_3$.  It has also been constructed numerically, and corresponds to a black hole solution that is also asymptotically locally AdS$_5$ with a squashed $S^3$ at infinity~\cite{Blazquez-Salcedo:2017kig, Blazquez-Salcedo:2017ghg}. However, since in this case $V(r)$ contains odd powers of $r$, by Corollary \ref{cor:u1su2BH} it does not have a smooth horizon.  In particular the function $V$ is $C^1$ but not $C^2$ as a function of $\lambda$ at the horizon, which implies the metric is $C^1$ but not $C^2$ at the horizon. To see this we calculate the invariant
\be
g_{33}= f^{-1} h_{33}- f^2 \omega_3^2= \ell^2 \left( \frac{455}{1936}-\frac{229 V_2 r^2}{2112}-\frac{1751  V_3 r^3}{6336}\right) + O(r^4) \;,
\ee
and  since $r= 2 a(\rho)= \sqrt{V_0} \rho (1+ O(\rho^2))$ the expansion in terms of $\rho$ takes the same form (up to unimportant numerical factors). Therefore using the coordinate change to Gaussian null coordinates (\ref{lambdarho}) we deduce that for any $V_3\neq 0$ the invariant $g_{33}:=g(L_3, L_3)$ as a function of $\lambda$ is $C^1$ but not $C^2$, as claimed.  An analogous class of black hole solutions with squashed $S^3$ boundary have been constructed in the more general $U(1)^3$ gauged supergravity~\cite{Cassani:2018mlh, Bombini:2019jhp}.  In fact, Propositions \ref{prop:NH} and \ref{prop:symBH} still apply in this more general theory (in particular, the near-horizon geometry has enhanced $U(1)
\times SU(2)$ symmetry), so our results also show that these solutions are $C^1$ but not $C^2$ at the horizon.

It should be noted that to rigorously prove that branch B and C solutions actually exist, one would have to prove that the series defined by (\ref{Vseries}) converges. Of course, one would like to relax the assumption of analyticity. It is natural to expect that the same result should be valid for smooth horizons. While parts of our arguments, in particular Proposition \ref{prop:symBH}, only require smoothness, our proof that (\ref{Vquad}) is the only solution to (\ref{ODE}) with the required boundary conditions uses analyticity in an essential way.  The status of this assumption therefore remains unclear.

\section{Discussion}

In this paper we have determined the general form for a timelike supersymmetric solution to five-dimensional minimal supergravity, under the additional assumption of an $SU(2)$-symmetry with 3d orbits.  This class is governed by a system of ODEs which in general do not appear integrable. However, under certain boundary conditions we found that the system implies an enhancement of symmetry to $U(1)\times SU(2)$ invariant solutions. This is a much simpler class that is governed by a  non-linear 5th order ODE of a single function. We also performed a near-horizon analysis for supersymmetric black holes with such an $SU(2)$ symmetry. This shows that a horizon manifests itself as a conical singularity in the K\"ahler base space (this also occurs for extremal black holes in flat-space~\cite{Reall:2002bh, Lucietti:2020ryg}). Crucially, this implies the required boundary conditions for the symmetry enhancement result so the solution must have $U(1)\times SU(2)$. Therefore the classification of black holes simplifies  and we were able to identify all analytic solutions to the 5th order ODE which governs this class, showing that the Gutowski-Reall black hole is the unique solution.

We emphasise that our symmetry assumptions are compatible with black holes with $S^3$ or lens space topology only. On the other hand, the only global assumption we make is that the supersymmetric Killing field is  timelike outside the horizon, so in particular our result applies equally to asymptotically globally AdS$_5$ and locally AdS$_5$ spacetimes. We have therefore shown that the Gutowski-Reall black hole is the only solution that is asymptotically globally AdS$_5$, in particular we deduce that there are no other solutions in this symmetry class  (either connected or disconnected to the Gutowski-Reall solution). Furthermore, we have also shown that there are no regular black hole solutions in this symmetry class that are asymptotically locally AdS$_5$ (at least assuming the metric at the horizon is analytic).

It is worth placing these results in the more general context of black hole classification.  In higher-dimensions the rigidity theorem implies that any stationary rotating black hole solution must have $U(1)^s$ rotational symmetry where $s \geq 1$~\cite{Hollands:2006rj, Hollands:2008wn}.\footnote{Technically, this has not been established for higher-dimensional extremal black holes in full generality.} In fact, this general result applies to  asymptotically flat and AdS spacetimes. Motivated by the rigidity theorem, it is expected that black holes with a single $U(1)$ symmetry may exist~\cite{Reall:2002bh}. On the other hand, in five-dimensions (asympotically flat, or globally AdS) all known explicit solutions possess the maximal abelian rotational symmetry $U(1)^2$ that is compatible with the $SO(4)$ rotation group at infinity.  Now, our assumption of $SU(2)$ symmetry generically only contains solutions with a $U(1)$ abelian rotational symmetry. Therefore, our black hole uniqueness theorem (in fact Proposition \ref{prop:symBH}), in this restricted context, implies the existence of a second commuting rotational symmetry, i.e. the abelian symmetry is enhanced to $U(1)^2\subset U(1) \times SU(2)$. While it is likely this is an artefact of the $SU(2)$ symmetry, it is an interesting open problem whether this enhancement of rotational symmetry is more generic or if there are supersymmetric black solutions with exactly $\mathbb{R}\times U(1)$ symmetry.  

The assumption of an $SU(2)$ rotational symmetry is expected to constrain the possible rotational configurations. Indeed, all known $D=5$ asymptotically flat/AdS black hole (and soliton) solutions with $SU(2)$ symmetry possess equal angular momenta $J_1=\pm J_2$ with respect to the orthogonal $U(1)^2$ Killing fields at infinity.  However, the converse statement need not be true, i.e. $J_1 = \pm J_2$ does not necessarily imply $SU(2)$ symmetry. In fact, for asymptotically flat solutions there are several known supersymmetric black holes with $J_1=\pm J_2$ that only have $U(1)^2$-symmetry, which explicitly shows that symmetry enhancement based on naive kinematics does not hold~\cite{Kunduri:2014iga, Horowitz:2017fyg, Breunholder:2018roc}.

As mentioned in the introduction a major open problem in this context is to also determine the (non)existence of black holes with non-spherical topology. Unfortunately, our strong symmetry assumption means the solutions are cohomogeneity-1, so does not allow us to address this question. In particular, it does not allow for multi-black holes so our work does not address their existence. Furthermore, while the symmetry assumption is compatible with lens space horizons, it does not allow for solutions with lens space horizons that are asymptotically globally AdS$_5$ (i.e. have a $S^3$ at infinity), and therefore we cannot address the existence of black lenses in this context.  In order to address these questions requires an analysis of cohomogeneity-2  solutions. We leave this to future work. \\

\noindent {\bf Acknowledgements}.  JL is supported by a Leverhulme Trust Research Project Grant. SO is supported by a Principal's Career Development Scholarship at the University of Edinburgh. We would like to thank Praxitelis Ntokos for helpful comments.

\appendix

\section{$SU(2)$ calculus}
\label{app:su2}
 
 Consider $G=SU(2)$ with the natural left and right $G$-action on $G$.
 Let $L_i$ be the generators of the left-action and $R_i$  be the generators of the right-action so,
 \be
 [L_i, L_j]= \epsilon_{ijk} L_k, \qquad [R_i, R_j]=-\epsilon_{ijk} R_k, \qquad [L_i, R_j]=0  \; ,
 \ee
 for $i, j, k= 1,2,3$. Thus $L_i$ ($R_i$) are right (left) invariant vector fields. The dual 1-forms, defined by $\sigma^L_i(L_j)=\delta_{ij}$ and $\sigma^R_i(R_j)=\delta_{ij}$, are right-invariant $\mathcal{L}_{R_i}\sigma^L_j=0$ and  left-invariant $\mathcal{L}_{L_i} \sigma^R_j=0$, and obey the Maurer-Cartan equations 
 \be
 \td \sigma^L_i=- \tfrac{1}{2} \epsilon_{ijk} \sigma^L_j \wedge \sigma^L_k, \qquad  
 \td \sigma^R_i= \tfrac{1}{2} \epsilon_{ijk} \sigma^R_j \wedge \sigma^R_k.
 \ee
 It follows that $\mathcal{L}_{{L}_i} \sigma^L_j = - \epsilon_{ijk} \sigma^L_k$ and $\mathcal{L}_{{R}_i} \sigma^R_j =  \epsilon_{ijk} \sigma^R_k$.
 
 It is useful to have an explicit coordinate system. We use Euler angles $(\theta, \phi, \psi)$, which (almost) cover $S^3\cong SU(2)$ if $0<\theta<\pi$,  $0<\phi<2\pi$, $0<\psi<4\pi$. In this coordinate system the right-invariant 1-forms can be written as
\bea
&&\sigma^L_1= \sin \psi \td \theta-\cos\psi\sin \theta \td \phi, \nonumber  \\   \nonumber
&&\sigma^L_2= \cos \psi \td \theta +\sin\psi\sin \theta \td \phi,\\ &&\sigma^L_3= \td \psi+ \cos\theta \td \phi  \; ,
\eea
and the dual vectors are
\bea
&&L_1= \cot \theta \cos\psi \partial_\psi+ \sin\psi \partial_\theta  -\frac{\cos\psi}{\sin\theta} \partial_\phi, \nonumber  \\  &&L_2= -\cot \theta \sin\psi \partial_\psi+ \cos\psi \partial_\theta  + \frac{\sin\psi}{\sin\theta} \partial_\phi, \nonumber \\  &&L_3 =\partial_\psi  \; .
\eea
The left-invariant 1-forms can be written as
\bea
&&\sigma^R_1=- \sin \phi \td \theta+\cos\phi\sin \theta \td \psi, \nonumber  \\   \nonumber
&&\sigma^R_2= \cos \phi \td \theta +\sin\phi\sin \theta \td \psi,\\ &&\sigma^R_3= \td \phi+ \cos\theta \td \psi  \; ,
\eea
and the dual vectors are
\bea
&&R_1=- \cot \theta \cos\phi \partial_\phi- \sin\phi \partial_\theta  +\frac{\cos\phi}{\sin\theta} \partial_\psi, \nonumber  \\  &&R_2= -\cot \theta \sin\phi \partial_\phi+ \cos\phi \partial_\theta  + \frac{\sin\phi}{\sin\theta} \partial_\psi, \nonumber \\  &&R_3 =\partial_\phi  \; .
\eea

 \section{Examples: K\"ahler metrics and solitons with a bolt}
 \label{app:bolts}
 
In this section we will study regularity of various examples of K\"ahler metrics with $U(1)\times SU(2)$ symmetry and of the associated supersymmetric solutions they define. We will assume the supersymmetric solutions have $U(1)\times SU(2)$ symmetry, i.e. $\omega_1=\omega_2=0$.

\subsection{K\"ahler-Einstein base}

The most general K\"ahler-Einstein metric (normalised so $R_{ab}=- \frac{6}{\ell^2} g_{ab}$) with $U(2)$ symmetry is given by 
\be
V(r)=1+ \frac{r^2}{\ell^2} +\frac{c_4}{r^4}  \; .
\ee
This  gives a supersymmetric solution with 
\be
f=1, \qquad \omega_3 = \frac{r^2}{2\ell} \; .
\ee
For $c_4=0$ this is the Bergmann metric and for $c_4>0$ there is a  singularity at $r=0$. 

For $c_4<0$, there is a largest real root $r_0>0$ such that $V(r_0)=0$.  One can write the solution in terms of $r_0$ as
\be
V(r)= \frac{(r^2-r_0^2)( r^2(r^2+\ell^2)+ r_0^2(r^2+r_0^2+\ell^2))}{\ell^2 r^4}  \; ,
\ee
which gives $r_0 V'(r_0)=4+ 6 r_0^2/\ell^2$. Therefore, regularity at $r=r_0$ requires it to be a bolt (\ref{EHbolt}) so
\be
p= 2+ \frac{3 r_0^2}{\ell^2}\; .
\ee
Note for $\ell\to \infty$ this reduces to $p=2$ as it should for the  Eguchi-Hanson metric.  However, for $\ell>0$ we see that there is a solution for every $p>2$ so we obtain corresponding smooth K\"ahler metrics.  Note that constant $r$ surfaces are squashed $S^3/\mathbb{Z}_p$.

In fact the corresponding supersymmetric solution is singular or has CTC. This is because regularity requires $\omega_3(r_0)=0$ which is never possible.  Alternatively, note that the Killing field $\partial_\psi - (r_0^2/(2\ell)) \partial_t$ has a fixed point in the spacetime at $r=r_0$ and so must have closed orbits which implies $t$ must be periodically identified.

\subsection{A class of soliton solutions}

Consider the ansatz
\be
V= c_0+ \frac{r^2}{\ell^2}+ \frac{c_2}{r^2} + \frac{c_4}{r^4}  \; .  \label{ansatz}
\ee
We find that this is a solution to (\ref{ODE}) iff
\be
c_2^2= 3 (c_0-1) c_4  \; .
\ee
The special case $c_0=1$ gives the K\"ahler-Einstein solution. These bases give
\bea
&&f = \frac{3 r^2}{(c_0-1)\ell^2+ 3 r^2}, \\
&&\omega_3=\frac{2 (c_0-1) c_2 \ell^4+(3 (c_0-1)^2 \ell^4 +9 c_2 \ell^2 )r^2+18 (c_0-1) \ell^2 r^4+18 r^6}{36 \ell r^4}  \; .
\eea
 All these solutions are asymptotically  locally AdS$_5$ in the sense that the K\"ahler base is asymptotically locally Bergmann, $f\sim 1$ and $\omega_3 \sim  r^2/(2\ell)$ as $r \to \infty$.

As $V(r)$ is not smooth at $r=0$ this solution cannot correspond to a black hole solution or a soliton with a nut (recall our earlier general analysis showed that in both cases smoothness implies $V(r)$ is a smooth positive function of $r^2$, see Section \ref{sec:NB} and Corollary \ref{cor:u1su2BH}). Therefore, the only possibility is that is has a bolt at $r=r_0>0$. 

To analyse this, it is convenient to change parameterisation and write
\be
V= \frac{(r^2-r_0^2) ( a_0+ a_1 r^2+ \frac{r^4}{\ell^2})}{r^4}  \; ,\label{alt}
\ee
so that one of our parameters is the root $r_0>0$. This is a solution to (\ref{ODE}) iff
\be
a_1^2 r_0^4+(a_1-3) a_0 r_0^2+a_0^2- \frac{3 a_0 r_0^4}{\ell^2}=0  \; .  \label{sol}
\ee
The regularity condition at the bolt (\ref{EHbolt}) fixes
\be
 a_1=p - \frac{r_0^2}{\ell^2} -\frac{a_0}{r_0^2}   \; ,
\ee
where recall $p\in \mathbb{N}$. These conditions give a family of complete K\"ahler metrics with a bolt.

Now consider a supersymmetric solution with such a base.  We must also require that $\omega$ is smooth at the bolt. In particular, this requires that $\omega_3(r_0)=0$. One finds that this gives another quadratic equation in $a_0, a_1$ and combining this with (\ref{sol}) implies
\be
a_0  = \frac{\ell^2 \left(2 p^2-4 p+3\right) r_0^2+(p-8) r_0^4}{\ell^2 (p+1)} \; .
\ee
Substituting back into (\ref{sol}) (or $\omega_3(r_0)=0$) finally gives
\be
27 x^2- (p-2) \left(p^2+14 p-5\right) x+(p-2)^3 p=0  \; ,
\ee
where we have defined the dimensionless parameter $x := r_0^2/\ell^2>0$. The two roots are given by
\be
x_\pm(p)= \frac{1}{54}(p-2) \left( p^2+14 p -5 \pm (1+p)\sqrt{(1+p)(25 +p)} \right)
\ee
and it can be shown that $x_+(p)>x_-(p)>0$ for all $p>2$, $x_-(1) \approx 0.08$ and $x_+(1)<0$ (clearly $p=2$ gives a trivial solution $x=0$).\footnote{For $r_0=0$ the above local solution has $a_0=0$ and reduces to the Gutowski-Reall case (Bergmann for $a_1=1$).} Each positive value of $x$ gives a smooth spacetime with a bolt at $r=r_0$  if $V>0$ for all $r>r_0$ and $f>0$ for $r \geq r_0$.

For $p=1$, the only allowed solution is $x_-(1)$ which does give $V>0$ for $r>r_0$, however in this case $f(r_0)<0$ so this does not give a soliton spacetime.  Similarly, for $p\geq 3$ the solution $x_-(p)$ always gives $f(r_0)<0$ and therefore must be discarded.

However,  for $p\geq 3$ one can show that the $x_+(p)$ solution gives $V>0$ for $r>r_0$  and $f>0$ for $r \geq r_0$. For example, for $p= 3$ and $x=x_+(3)= \left(8 \sqrt{7}+23\right)/27$ we get the solution
\be
V= (r^2-r_0^2)\left( \frac{1}{\ell^2}+\frac{2(13+\sqrt{7})}{27 r^2} + \frac{2 \left(13 \sqrt{7}+88\right) \ell^2}{729 r^4} \right)
\ee
and
\bea
f^{-1} &=& \frac{r^2-r_0^2+\frac{1}{9} \left(2 \sqrt{7}+5\right) \ell^2}{r^2} \\
\omega_3 &=& (r^2-r_0^2)\left( -\frac{2 \left(7 \sqrt{7}+10\right) \ell^3}{729 r^4}+\frac{\left(2 \sqrt{7}-1\right) \ell}{54 r^2}+\frac{1}{2 \ell}\right) 
\eea
which manifestly satisfies $V>0$ for $r>r_0$ and $f>0$ for $r\geq r_0$. This gives an asymptotically  AdS$_5/ \mathbb{Z}_3$ soliton with a bolt at $r=r_0$ and a spatial $S^3/\mathbb{Z}_3$ boundary metric.  More generally, for all $p \geq 3$ the solution $x_+(p)$ gives asymptotically AdS$_5/ \mathbb{Z}_p$ solitons with a bolt with spatial $S^3/\mathbb{Z}_p$ boundary metric.

\subsection{Extremal K\"ahler metrics}

Although out of the main line of this paper, we have obtained a class of extremal K\"ahler metrics analogous to the solution above.
The equation for an extremal K\"ahler metric is~\cite{extremal}
\be
 \nabla^c \left( \nabla_c \nabla^2 R+ 2R_{cb}\nabla^b R \right)=0 \; .  \label{extremal}
\ee
This in fact comprises two of the terms in the equation for supersymmetry (\ref{PDE}), although any relation remains unclear.

In any case we can find a family of $U(1)\times SU(2)$ invariant extremal K\"ahler metrics (\ref{EHgauge}) again using the ansatz  (\ref{ansatz}). We now find this is extremal iff  $c_0=1$. Thus
\be
V= 1+\frac{r^2}{\ell^2} + \frac{c_2}{r^2}+\frac{c_4}{r^4}
\ee
gives a 2-parameter family of extremal K\"ahler metrics. For $c_2=0$ this is K\"ahler-Einstein.  In general we find
\be
R= -\frac{24}{\ell^2}
\ee 
so these are constant scalar curvature (note both $c_2, c_4$ do not appear in the scalar curvature).  In fact, the most general $U(1)\times SU(2)$ K\"ahler metric with constant scalar curvature $R=-24/ \ell^2$  is given by the above solution.  We emphasise that these are not solutions to (\ref{ODE}) unless $c_2=0$ and therefore do not give rise to supersymmetric solutions.

To analyse regularity it is convenient to use the alternate parameterisation (\ref{alt}). The extremal condition is now 
\be
a_1=1+\frac{r_0^2}{\ell^2}
\ee
and the regularity condition (\ref{EHbolt}) at the bolt implies
\be
a_0= -r_0^2 \left(1-p+\frac{2r_0^2}{\ell^2} \right)
\ee
where $p\in \mathbb{N}$. This gives
\be
V= \frac{(r^2-r_0^2)}{r^4} \left(\frac{(r^2+2r_0^2) (r^2-r_0^2)}{\ell^2} + r^2 + r_0^2(p-1)\right)  \; .
\ee
Clearly $V(r)>0$ for all $r>r_0$ so this is  a globally regular metric for any $r_0>0$ and any $p\in \mathbb{N}$. These are complex-line bundles over $S^2$ associated to circle-bundles over $S^2$ with Chern number $p$. For $\ell \to \infty$ these metrics coincide with the scalar flat K\"ahler metrics of LeBrun (his $a=r_0$ and his $k=p-1$)~\cite{lebrun}.

\end{document}